\documentclass[runningheads]{llncs}
\usepackage[T1]{fontenc}
\usepackage{graphicx}

\usepackage{enumitem}
\usepackage{subcaption}
\usepackage{amsmath}
\usepackage{amssymb}
\usepackage{mathtools}
\usepackage[dvipsnames,table,xcdraw]{xcolor}
\usepackage{complexity}
\usepackage{tikz} 
    \usetikzlibrary{automata, positioning, arrows, decorations.pathreplacing}
\usepackage{todonotes}
\usepackage{hyperref}
\usepackage[capitalise]{cleveref}

\newcommand{\tup}[1]{\langle #1 \rangle}

\newcommand{\cA}{\mathcal{A}}
\newcommand{\cB}{\mathcal{B}}
\newcommand{\cC}{\mathcal{C}}
\newcommand{\cD}{\mathcal{D}}

\newcommand{\cU}{\mathcal{U}}
\newcommand{\cV}{\mathcal{V}}
\renewcommand{\cP}{\mathcal{P}}

\newcommand{\bbN}{\mathbb{N}}
\newcommand{\bbZ}{\mathbb{Z}}

\renewcommand{\lang}{\mathcal{L}}
\newcommand{\effect}[1]{\textsf{eff}(#1)}

\definecolor{cbOne}{HTML}{DC3220} %colorblind1
\definecolor{cbTwo}{HTML}{005AB5} %colorblind2
\definecolor{cbThree}{HTML}{183E1C} %colorblind2

\newcommand{\positive}{{\color{cbOne}>0}}
\newcommand{\nonnegative}{{\color{cbTwo} \geq 0}}
\newcommand{\negative}{{\color{cbThree} < 0}}

\newcommand{\posCol}[1]{{\color{cbOne}{#1}}}
\newcommand{\nonnegCol}[1]{{\color{cbTwo}{#1}}}

\newcommand{\newl}{L}%{{\rm new\_}l}
\newcommand{\newb}{B}%{{\rm new\_}b}

\renewcommand{\vec}[1]{\boldsymbol{#1}}

\renewcommand{\phi}{\varphi}

\newcommand{\set}[1]{\{#1\}}

\newcommand{\Uu}{\mathcal{U}}

\newcommand{\abs}[1]{|#1|}

\crefname{conjecture}{Conjecture}{Conjectures} % For cref

\begin{document}

\title{Dimension-Minimality and Primality of Counter Nets\thanks{S. Almagor was supported by the ISRAEL SCIENCE FOUNDATION (grant No. 989/22), G. Avni was supported by the ISRAEL SCIENCE FOUNDATION (grant No. 1679/21), H. Sinclair-Banks was supported by EPSRC Standard Research Studentship (DTP), grant number EP/T5179X/1.}}
%
%\titlerunning{Abbreviated paper title}
% If the paper title is too long for the running head, you can set
% an abbreviated paper title here
%
\author{Shaull Almagor\inst{1}\orcidID{0000-0001-9021-1175} \and Guy Avni\inst{2}\orcidID{0000-0001-5588-8287}  \and \\ Henry Sinclair-Banks\inst{3}\orcidID{0000-0003-1653-4069} \and Asaf Yeshurun\inst{1}}
\authorrunning{S. Almagor et al.}
% First names are abbreviated in the running head.
% If there are more than two authors, 'et al.' is used.
%
\institute{Technion, Israel \\ \email{shaull@technion.ac.il,asafyeshurun@campus.technion.ac.il} 
\and Department of Computer Science, University of Haifa, Israel \\ \email{gavni@cs.haifa.ac.il} \and Centre for Discrete Mathematics and its Applications (DIMAP) \&\\ Department of Computer Science, University of Warwick, Coventry, UK\\ \email{h.sinclair-banks@warwick.ac.uk}
} 
\maketitle              % typeset the header of the contribution

%TODO mandatory: add short abstract of the document
\begin{abstract}
    
A $k$-Counter Net ($k$-CN) is a finite-state automaton equipped with $k$ integer counters that are not allowed to become negative, but do not have explicit zero tests. 
This language-recognition model can be thought of as labelled vector addition systems with states, some of which are accepting.
Certain decision problems for $k$-CNs become easier, or indeed decidable, when the dimension $k$ is small.
Yet, little is known about the effect that the dimension $k$ has on the class of languages recognised by $k$-CNs. 
Specifically, it would be useful if we could simplify algorithmic reasoning by reducing the dimension of a given CN.

To this end, we introduce the notion of dimension-primality for $k$-CN, whereby a $k$-CN is prime if it recognises a language that cannot be decomposed into a finite intersection of languages recognised by $d$-CNs, for some $d<k$. 
We show that primality is undecidable.
We also study two related notions: dimension-minimality (where we seek a single language-equivalent $d$-CN of lower dimension) and language regularity.
Additionally, we explore the trade-offs in expressiveness between dimension and non-determinism for CN.

\end{abstract}

% \section{Paper structure suggestion}
% \begin{enumerate}
%     \item Intro - discuss primality and the three parameters (dimension, number, size), and that we only talk about the first two. Also discuss Regularity as an almost ``degenerate'' case. Have an example?
%     \item Prelim (defining $k$-CN and $k$-DCN, defining projections on dimensions)
%     \item Primality and dimension minimization: define primality with the two parameters above. Perhaps write that a $d$-CN is Prime$(k,m)$ if it cannot be written as an intersection of at most $m$ $k$-CNs, and call it Prime if it is Prime$(k,m)$ for all $k<d$ and all $m\in \bbN$.
    
%     Give basic result about $k$-DCN and projections, and show that for $k$-DCN minimization and primality are the same. Perhaps discuss unambiguity.
%     \item A Prime 2-CN: example and proof.
%     \item Undecidability of Primality and of Minimization. Reduction using the example to minimization and to primality.
%     \item Regularity of $d$-DCN is decidable
%     \item Examples (possibly with a better section name). Showing the various separation results we have.
%     \item Discussion and Open Problems
% \end{enumerate}
\section{Introduction}
\label{sec:introduction}
A \emph{$k$-dimensional Counter Net} ($k$-CN) is a finite-state automaton equipped with $k$ integer counters that are not allowed to become negative, but do not have explicit zero tests (see~\cref{subfig:2CN_intro_a} for an example). 
This language-recognition model can be thought of as an alphabet-labelled Vector Addition System with States (VASS), some of whose states are accepting~\cite{czerwinski2020universality}. 
A $k$-CN $\cA$ over alphabet $\Sigma$ \emph{accepts} a word $w\in \Sigma^*$ if there is a run of $\cA$ 
on $w$ that ends in an accepting state in which the counters stay non-negative. The \emph{language} of $\cA$ is the set $\lang(\cA)$ of words accepted by $\cA$.

%$k$-CNs are a natural model of concurrency and are closely related --- and equivalent, in some senses --- to labelled Petri Nets. 
Counter nets are a natural model of concurrency and are closely related --- and equivalent, in some senses --- to labelled Petri Nets. 
These models have received significant attention over the years~\cite{cabasino2013diagnosability,czerwinski2020universality,esparza2005decidability,figueira2019co,greibach1978remarks,hack1976petri,render2009rational}, with specific interest in the one-dimensional case, often referred to as one-counter nets~\cite{hofman2013decidability,hofman2014trace,almagor2020parametrized,almagor2022determinization}. 
Unfortunately, most decision problems for $k$-CNs are notoriously difficult and are often undecidable~\cite{almagor2020parametrized,almagor2022determinization}. In particular, $k$-CNs subsume VASS and Petri nets, for which many problems are known to be Ackermann-complete, for example see the recent breakthrough in the complexity of reachability in VASS~\cite{czerwinski2022reachability,leroux2022reachability}.

In many cases, the complexity of decision problems for VASS, sometimes with extensions, depends on the dimension, with low dimensions admitting more tractable solutions.~\cite{czerwinski2019new,czerwinski2020reachability,czerwinski2021improved,finkel2018reachability}. 
%
%In recent years it was noticed that for low dimensions, decision problems for VASS, sometimes with extensions, become more tractable~\cite{czerwinski2019new,czerwinski2020reachability,czerwinski2021improved,finkel2018reachability}. 
For example, reachability in dimensions one and two is \NP-complete~\cite{haase2009reachability} and \PSPACE-complete~\cite{blondin2021reachability}, respectively, when counter updates are encoded in binary.

A natural question, therefore, is whether we can \emph{decrease} the dimension of a given a $k$-CN whilst maintaining its language, to facilitate reasoning about it. More generally, the trade-off between expressiveness and the dimension of Counter Nets is poorly understood. 
%Since $k$-CNs are used for various purposes (usually in the form of Petri nets), a positive answer to this question could allow us to models systems of size previously untractable using richer parallelism. 
We tackle this question in this work by introducing two approaches.
The first is straightforward \emph{dimension-minimality}: given a $k$-CN, does there exist a $d$-CN $\cB$ recognising the same language for some $d<k$? 

The second approach is \emph{primality}: given a $k$-CN, does there exist some $d<k$ and $d$-CNs $\cB_1, \ldots, \cB_n$ such that $L(\cA)=\bigcap_{i=1}^n \lang(\cB_i)$?
That is, we ask whether the language of $\cA$ can be decomposed as an intersection of languages recognised by several lower-dimension CNs. 
We also consider \emph{compositeness}, the dual of primality.
Intuitively, in a composite $k$-CN the usage of the counters can be ``split'' across several lower-dimension CNs, allowing for properties (such as universality) to be checked on each conjunct separately. 
%Guy We illustrate this definition in~\cref{fig:composite_2CN_intro}.

\begin{example}
\label{xmp:composite 2CN}
We illustrate the model and the definition of compositionality. Consider the $2$-CN $\cA$ depicted in \cref{subfig:2CN_intro_a}, and consider a word $w = a^m \# b^n \# c^k$. We have that $\cA$ has an accepting run on $w$ iff $m \geq n$ and $m \geq k$. Indeed, if $m < n$, the first counter drops below $0$ while cycling in the second state and so the run is ``stuck'', and similarly if $m<k$. It is not hard to show that there is no $1$-CN that recognizes the languages of $\cA$. However, \cref{subfig:2CN_intro_b} shows two $1$-CNs $\cB_1$ and $\cB_2$ such that $\lang(\cB) = \lang(\cB_1) \cap \lang(\cB_2)$. Indeed, a word $w = a^m \# b^n \# c^k \in \lang(\cB_1)$ iff $m \geq n$, and $w \in \lang(\cB_2)$ iff $m \geq k$. 
%Note that this decomposition is obtained by automata with the same structure as $\cA$ and by ``splitting'' the counters between the two parts. However, as we show in \cref{thm:2CN_prime}, such a partition of counters does not in general yield a decomposition, but does work for deterministic $1$-CNs, as we show in~\cref{prop:DCN_Projection}.
\end{example}

\begin{figure}[ht]
    \centering
     \begin{subfigure}{0.35\textwidth}     
        \begin{tikzpicture}[auto,node distance=1.5cm,scale=1]
        \footnotesize
            \node (q0) [initial, state, initial text = {},inner sep=3pt, minimum size=5pt] {};
            \node (q1) [state,inner sep=3pt, minimum size=3pt] at (1.6,0) {};
            \node (q2) [accepting, state,inner sep=3pt, minimum size=5pt] at (3.2,0) {};
            \path [-stealth]
            (q0) edge [bend left=0] node[below] {$\#,(0,0)$} (q1)
            (q1) edge [bend left=0] node[below] {$\#,(0,0)$} (q2)
            (q0) edge [loop above] node {$a,(1,1)$} (q0)
            (q1) edge [loop above] node {$b,(-1,0)$} (q1)
            (q2) edge [loop above] node {$c,(0,-1)$} (q2);
            %(q1) edge [in=100,out=135,loop] node {$a$} (q1)
        \end{tikzpicture}
        \caption{A composite $2$-CN.}
        \label{subfig:2CN_intro_a}
        \end{subfigure}
     \hfill
     \begin{subfigure}{0.64\textwidth}
    \centering
    \begin{tikzpicture}[auto,node distance=1.5cm,scale=1]
        \footnotesize
            \node (q0) [initial, state, initial text = {},inner sep=3pt, minimum size=5pt] {};
            \node (q1) [state,inner sep=3pt, minimum size=3pt] at (1.5,0) {};
            \node (q2) [accepting, state,inner sep=3pt, minimum size=5pt] at (3,0) {};
            \path [-stealth]
            (q0) edge [bend left=0] node[below] {$\#,(0)$} (q1)
            (q1) edge [bend left=0] node[below] {$\#,(0)$} (q2)
            (q0) edge [loop above] node {$a,(1)$} (q0)
            (q1) edge [loop above] node {$b,(-1)$} (q1)
            (q2) edge [loop above] node {$c,(0)$} (q2);
            %(q1) edge [in=100,out=135,loop] node {$a$} (q1)
        \end{tikzpicture}
        \begin{tikzpicture}[auto,node distance=1.5cm,scale=1]
        \footnotesize
            \node (q0) [initial, state, initial text = {},inner sep=3pt, minimum size=5pt] {};
            \node (q1) [state,inner sep=3pt, minimum size=3pt] at (1.5,0) {};
            \node (q2) [accepting, state,inner sep=3pt, minimum size=5pt] at (3,0) {};
            \path [-stealth]
            (q0) edge [bend left=0] node[below] {$\#,(0)$} (q1)
            (q1) edge [bend left=0] node[below] {$\#,(0)$} (q2)
            (q0) edge [loop above] node {$a,(1)$} (q0)
            (q1) edge [loop above] node {$b,(0)$} (q1)
            (q2) edge [loop above] node {$c,(-1)$} (q2);
            %(q1) edge [in=100,out=135,loop] node {$a$} (q1)
        \end{tikzpicture}
        \caption{Two $1$-CNs showing compositeness of the $2$-CN.}
        \label{subfig:2CN_intro_b}
     \end{subfigure}
     \caption{A composite $2$-CN whose language is $\{a^m\#b^n\#c^k\mid m\ge n\wedge m\ge k\}$ and its decomposition into two $1$-CNs recognising the languages $\{a^m\#b^n\#c^k\mid m\ge n\}$ and $\{a^m\#b^n\#c^k\mid m\ge k\}$.\\
     %Guy Note that for this decomposition we merely ``separate'' the counters. However, separating the counters does not, in general, yield a decomposition (c.f., \cref{thm:2CN_prime}).
     }
    \label{fig:composite_2CN_intro}
\end{figure}

Note that the decomposition in \cref{xmp:composite 2CN} is obtained by ``splitting'' the counters between the two $1$-CNs. This raises the question of whether such splittings are always possible. 
As we show in~\cref{prop:DCN_Projection}, for deterministic $k$-CNs ($k$-DCNs) this is indeed the case. In general, however, it is not hard to find examples where a $k$-CN cannot simply be split to an intersection by projecting on each counter. This however, does not rule out that other decompositions are possible.
Our main result,~\cref{thm:2CN_prime}, gives an example of a prime $2$-CN. That is, a $2$-CN whose language cannot be expressed as an intersection of $1$-CNs.

The notion of primality has been studied for regular languages in~\cite{kupferman2015prime,jecker2020unary,jecker2021decomposing}, the exact complexity of deciding primality is still open. 
There, an automaton is composite if it can be written as an intersection of finite automata with fewer \emph{states}. 
In this work we introduce primality for CNs.
We focus on \emph{dimension} as a measure of size, a notion which does not exist for regular languages. 
Thus, unlike regular languages, the differences between prime and composite CNs is not only in succinctness, but actually in expressiveness, as we later demonstrate. 

We parameterise primality and compositeness by the dimension $d$ and the number $n$ of lower-dimension factors. 
Thus, a $k$-CN $\cA$ is \emph{$(d,n)$-composite} if it can be written as the intersection above. 
Then, $\cA$ is \emph{composite} if it is $(d,n)$-composite for some $d<k$ and $n\in \bbN$. 
Under this view, dimension-minimality is a special case of compositeness, namely $\cA$ is dimension-minimal if it is not $(k-1,1)$-composite. 
Another particular problem captured by compositeness is \emph{regularity}.
Indeed, $\lang(\cA)$ is regular if and only if $\cA$ is $(0,1)$-composite, since $0$-CNs are just NFAs.
Since regularity is already undecidable for $1$-CNs~\cite{almagor2022determinization,Valk81}, it follows that deciding whether a $k$-CN is $(d,n)$-composite is undecidable. 
Moreover, it follows that both primality and dimension-minimality are undecidable for $1$-CNs. 

The undecidability of the above problems is not surprising, as the huge difference in expressive power between $1$-CNs and regular languages is well understood. 
In contrast, even the expressive power difference between $1$-CNs and $2$-CNs is poorly understood, let alone what effect the dimension has on the expressive power beyond regular languages.
Already, 1-VASS and 2-VASS are known to have \emph{flat} equivalents with respect to reachability~\cite{leroux2004flatness,blondin2021reachability}, but the complexity differs greatly.

Our goal in this work is to shed light on these differences. 
In~\cref{sec:prime-2vass}, we give a concrete example of a prime $2$-CN, which turns out to be technically challenging. 
This example is the heart of our technical contribution, and we emphasise that we \emph{do not} know have a proved example of a prime $3$-CN, let alone for general $k$-CN (although we conjecture a candidate for such languages). We consider this an interesting open problem, as it highlights the type of pumping machinery that is currently missing from the VASS/CN reasoning arsenal.
The technical intricacy in proving our example suggests that generalising it is highly nontrivial. Indeed, proving this claim would require intricate pumping arguments, which are notoriously difficult even for low-dimensional CNs~\cite{czerwinski2019new}.

Using our example, we obtain in~\cref{sec:primality_undecidable}, the undecidability of primality and of dimension-minimality for $2$-CNs. 
To complement this, we show in~\cref{thm:DCN_regularity_decidable}, that regularity of $k$-DCNs is decidable.
In~\cref{sec:dimension_vs_nondet}, we explore trade-offs in expressiveness of CNs with increasing dimension and with nondeterminism. In particular, we show that there is a strict hierarchy of expressiveness with respect to the dimension. 
We conclude with a discussion in~\cref{sec:discussion}.
For brevity, some proofs appear in the appendix.

\section{Preliminaries}
\label{sec:preliminaries}
We denote the non-negative integers $\{0,1,\ldots\}$ by $\bbN$. 
We write vectors in bold, e.g., $\vec{e}\in \bbZ^k$, and $\vec{e}[i]$ is the $i$-th coordinate. 
We use $[k]=\{1,\ldots,k\}$ for $k\ge 1$. 
We use $\Sigma^*$ to denote the set of all words over an alphabet $\Sigma$, and $|w|$ is the length of $w\in \Sigma^*$.

A \emph{$k$-dimensional Counter Net} ($k$-CN) $\cA$ is a quintuple $\cA=\tup{\Sigma, Q, Q_0, \delta, F}$ where  $\Sigma$ is a finite alphabet, $Q$ is a finite set of states, $Q_0 \subseteq Q$ is the set of initial states, $\delta \subseteq Q \times \Sigma \times \bbZ^k \times Q$ is a set of transitions, and $F \subseteq Q$ are the accepting states. 
A $k$-CN is \emph{deterministic}, denoted $k$-DCN, if $|Q_0| = 1$, and for every $p \in Q$ and $\sigma \in \Sigma$ there is at most one transition of the form $(p, \sigma, \vec{v}, q) \in \delta$. 
For a transition $(p, \sigma, \vec{v}, q)\in\delta$, we refer to $\vec{v} \in \bbZ^k$ as its \emph{effect}. 

An $\bbN$-configuration (resp. $\bbZ$-configuration) of a $k$-CN $\cA$ is a pair $(q,\vec{v})\in Q\times \bbN^k$ (resp. $(q,\vec{v})\in Q\times \bbZ^k$) representing the current state and values of the counters.
A transition $(p, \sigma, \vec{e}, q) \in \delta$ is \emph{valid} from $\bbN$-configuration $(q, \vec{v})$ if $\vec{v} + \vec{e} \in \bbN^k$, i.e., if all $k$ counters remain non-negative after the transition. 
A \emph{$\bbZ$-run} $\rho$ of $\cA$ on $w$ is a sequence of $\bbZ$-configurations $\rho=(q_0, \vec{v}_0), (q_1, \vec{v}_1), \ldots, (q_n, \vec{v}_n)$ such that $(q_i, \sigma_i, \vec{v}_{i+1} - \vec{v}_i, q_{i+1}) \in \delta$ for every $0 \leq i \leq n-1$, we may also say that $\rho$ \emph{reads} $w = \sigma_0 \sigma_1 \cdots \sigma_n$.
An \emph{$\bbN$-run} is a $\bbZ$-run that visits only $\bbN$-configuration.
Note that all the transitions in an $\bbN$-run are valid. 
We may omit $\bbN$ or $\bbZ$ from the run when it does not matter.
For a run $\rho = (q_0, \vec{v}_0), (q_1, \vec{v}_1), \ldots, (q_n, \vec{v}_n)$ of $\cA$, we denote $(q_0,\vec{v}_0)\stackrel{\rho}{\to}(q_n, \vec{v}_n)$. 
We define the \emph{effect} of $\rho$ to be $\effect{\rho}=\vec{v}_n - \vec{v}_0$. 

An $\bbN$-run $\rho$ is \emph{accepting} if $q_0 \in Q_0$, $\vec{v}_0 = \vec{0}$, and $q_n \in F$. 
We say that $\cA$ \emph{accepts} $w$ if there is an accepting $\bbN$-run of $\cA$ on $w$. 
The \emph{language} of $\cA$ is $\lang(\cA)=\{w\in \Sigma^* \mid \cA \text{ accepts } w\}$.   

An infix $\pi = (q_k, \vec{v}_k), (q_{k+1}, \vec{v}_{k+1}), \ldots, (q_{k+n}, \vec{v}_{k+n})$ of a run $\rho$ is a \emph{cycle} if $q_k = q_{k+n}$ and is a \emph{simple cycle} if it does not contain a cycle as a proper infix.  
When discussing an infix $\pi$ of a $1$-CN -- we write that $\pi$ is $\positive$, $\nonnegative$, or $\negative$ if $\effect{\pi} > 0$, $\effect{\pi} \geq 0$, or $\effect{\pi} < 0$, respectively. 

\section{Primality and Compositeness}
\label{sec:primality}

We begin by presenting our main definitions, followed by some introductory properties.

\begin{definition}[Compositeness, Primality, and Dimension-Minimality]
    \label{def:primality}
    % Consider a $k$-CN $\cA$, and let $d,n\in \bbN$. 
    % If there exist $d$-CNs $\cB_1,\ldots,\cB_n$ such that $\lang(\cA) = \bigcap_{i=1}^n \lang(\cB_i)$, then $\cA$ is \emph{$(d,n)$-composite}.
    % If $\cA$ is $(d,n)$-composite for some $d<k$ and $n\in \bbN$, then $\cA$ is \emph{composite}.
    % Otherwise, $\cA$ is \emph{prime}.
    % If $\cA$ is not $(k-1,1)$-composite, then $\cA$ is \emph{dimension-minimal}.
    Consider a $k$-CN $\cA$, and let $d,n\in \bbN$. We say that $\cA$ is \emph{$(d,n)$-composite} if there exist $d$-CNs $\cB_1,\ldots,\cB_n$ such that $\lang(\cA)=\bigcap_{i=1}^n\lang(\cB_i)$.	
If $\cA$ is $(d,n)$-composite for some $d<k$ and $n\in \bbN$, we say $\cA$ is \emph{composite}. Otherwise,  $\cA$ is \emph{prime}.	
If $\cA$ is not $(k-1,1)$-composite, we say that $\cA$ is \emph{dimension-minimal}
\end{definition}
\begin{remark}
    \label{rmk:compositeness_generalizes_reg}
    Note that the special case where $\cA$ is $(0,n)$-composite coincides with the regularity of $\lang(\cA)$, and hence also with being $(0,1)$-composite.
\end{remark}

Observe that in~\cref{fig:composite_2CN_intro} we in fact show a composite $2$-DCN. 
We now show that every $k$-DCN is $(1,k)$-composite, by projecting to each of the counters separately. 
In particular, a $k$-DCN is prime only when $k=1$ and it recognises a non-regular language, or when $k=0$.
Formally, consider a $k$-DCN $\cD=\tup{\Sigma,Q,Q_0,\delta,F}$ and let $1\le i\le k$. 
We define the \emph{$i$-projection} to be the 1-DCN $\cD|_i=\tup{\Sigma,Q,Q_0,\delta|_i,F}$ where $\delta|_i = \set{ (p, \sigma, \vec{v}[i], q) \mid (p, \sigma, \vec{v}, q) \in \delta }$. 

\begin{proposition}%[Deterministic Counter Nets are Composite]
    \label{prop:DCN_Projection}
    Every $k$-DCN $\cD$ is $(1,k)$-composite. 
    Moreover, $\lang(\cD)=\bigcap_{i=1}^k\lang(\cD|_i)$.
\end{proposition}
\begin{proof}
Let $w\in \lang(\cD)$ and let $\rho$ be the accepting run of $\cD$ on $w$, then the projection of $\rho$ on counter $i$ induces an accepting run of $\cD|_i$ on $w$, thus $w\in \bigcap_{i=1}^k\lang(\cD|_i)$. 
Note that this direction does not use the determinism of $\cD$.

Conversely, let $w\in \bigcap_{i=1}^k\lang(\cD|_i)$, then each $\cD|_i$ has an accepting run $\rho_i$ on $w$. 
Since the structure of all the $\cD|_i$ is identical to that of $\cD$, all the runs $\rho_i$ have identical state sequences, and therefore are also a $\bbZ$-run of $\cD$ on $w$. 
Moreover, due to this being a single $\bbN$-run in each $\cD|_i$, it follows that all counter values remain non-negative in the corresponding run of $\cD$ on $w$. 
Hence, this is an accepting $\bbN$-run of $\cD$ on $w$, so $w\in \lang(\cD)$.
\hfill\qed\end{proof}

\begin{remark}[Unambiguous Counter Nets are Composite]
\label{rmk:unambiguous}
The proof of~\cref{prop:DCN_Projection} applies also to \emph{structurally unambiguous} CNs, i.e. CNs whose underlying automaton, disregarding the counters, is unambiguous.
Thus, every unambiguous CN is $(1, k)$-composite.
\end{remark}

Consider $k$-CNs $\cB_1,\ldots,\cB_n$. 
By taking their product, we can construct a $kn$-CN $\cA$ such that $\lang(\cA) = \bigcap_{i=1}^n \lang(\cB_i)$. 
In particular, if each $\cB_i$ is a $1$-DCN, then $\cA$ is an $n$-DCN. 
Combining this with~\cref{prop:DCN_Projection}, we can deduce the following (proof in~\cref{apx:DCN_dimension_minimal_iff_composite}).
\begin{proposition}
\label{prop:DCN_dimension_minimal_iff_composite}
    A $k$-DCN is dimension-minimal if and only if it is not $(1,k-1)$-composite.
\end{proposition}

\section{A Prime Two-Counter Net}
\label{sec:prime-2vass}
In this section we present our main technical contribution, namely an example of a prime $2$-CN. 
The technical difficulty arises from the need to prove that this example cannot be decomposed as an \emph{intersection} of \emph{nondeterministic} $1$-CNs. 
Since intersection has a ``universal flavour'', and nondeterminism has an ``existential flavour'', we have a sort of ``quantifier alternation'' which is often a source of difficulty.

The importance of this example is threefold.
First, it enables us to show in that primality is undecidable in~\cref{sec:primality_undecidable}. 
Second, it offers intuition on what makes a language prime. 
Third, we suspect that the techniques developed here will be useful in other settings when reasoning about nondeterministic automata, perhaps with counters.

We start by presenting the prime $2$-CN, followed by an overview of the proof, before delving into the details.
\begin{example}
\label{xmp:prime2CN}
Consider the $2$-CN $\cP$ over alphabet $\Sigma=\{a,b,c,\#\}$ depicted in~\cref{fig:prime-2vass}.
Intuitively, $\cP$ starts by reading segments of the form $a^m\#$, where in each segment it nondeterministically chooses whether to increase the first or second counter by $m$. 
Then, it reads $b^{m_b}c^{m_c}$ and accepts if the value of the first and second counter is at least $m_b$ and $m_c$, respectively.
Thus, $\cP$ accepts a word if its $a^m\#$ segments can be partitioned into two sets $I$ and $\overline{I}$ so that the combined lengths of the segments in $I$ (resp. $\overline{I}$) is at least the length of the $b$ segment (resp. $c$ segment).		
For example, $a^{10}\#a^{20}\#a^{15}\#b^{15}c^{30}\in \lang(\cP)$, since segments 1 and 2 have length $30$, matching $c^{30}$ and segment $3$ matches $b^{15}$. However, $a^{10}\#a^{20}\#a^{15}\#b^{21}c^{21}\notin \lang(\cP)$, since in any partition of $\{10,20,15\}$, one set will have sum lower than $21$.		
More precisely, we have the following:	
\begin{equation*}
    \lang(\cP)=\{a^{m_1}\#a^{m_2}\#\cdots\#a^{m_{t}}\#b^{m_b}c^{m_c} \mid \exists I\subseteq [t] \text{ s.t. }\sum_{i\in I}m_i\ge m_b \wedge \sum_{i\notin I}m_i\ge m_c\}
\end{equation*}
\begin{figure}[ht!]
    \centering
    %%%use fbox to make sure we're not wasting vertical space.
    %\fbox{
    \begin{tikzpicture}[auto,node distance=1.5cm,scale=1]
        \footnotesize
        \thinmuskip=0mu
            \clip (-4,-1.3) rectangle (7, 1.3);  
            \node (p) [initial, state, initial text = {},inner sep=3pt, minimum size=5pt] at (0,1) {};
            \node (q) [initial, state, initial text = {}, inner sep=3pt, minimum size=3pt] at (0,-1) {};
            \node (r) [accepting, state,inner sep=3pt, minimum size=5pt] at (3,0) {};
            \node (f) [accepting, state,inner sep=3pt, minimum size=5pt] at (5.5,0) {};
            \path [-stealth]
            (p) edge [bend left=20] node[right] {$\#,(0,0)$} (q)
            (q) edge [bend left=20] node[left] {$\#,(0,0)$} (p)
            (p) edge [loop, out=220, in=150,distance = 15mm,] node[left=-2mm] {$\begin{array}{r}
                   a,(1,0)\\ \#,(0,0)
            \end{array}$} (p)
            (q) edge [loop, out=210, in=140,distance = 15mm] node[left=-2mm] {$\begin{array}{r}
                   a,(0,1)\\ \#,(0,0)
            \end{array}$} (q)
            (p) edge [bend left=20] node[below, sloped, pos = 0.4] {$\#,(0,0)$} (r)
            (q) edge [bend right=20] node[above, sloped, pos = 0.4] {$\#,(0,0)$} (r)
            (r) edge [loop above,distance=10mm,out=120,in=60] node {$b,(-1,0)$} (r)
            (f) edge [loop above,distance=10mm,out=120,in=60] node {$c,(0,-1)$} (f)
            (r) edge [bend left=0] node {$c,(0,-1)$} (f);
\end{tikzpicture}
    %}
    \caption{The prime $2$-CN $\cP$ for~\cref{xmp:prime2CN} and~\cref{thm:2CN_prime}.}
    \label{fig:prime-2vass}
\end{figure}
\end{example}

\begin{theorem}
\label{thm:2CN_prime}
$\cP$ is prime. 
\end{theorem}
The high-level intuition behind~\cref{thm:2CN_prime} is that any $1$-CN can either guess a subset of segments that covers $m_b$ or $m_c$, but not both, and in order to make sure the choices between two $1$-CNs form a partition, we need to fix the partition in advance. This is only possible if the number of segments is a priori fixed, which is not true (c.f.,~\cref{rmk:unbounded_compositeness}). This intuition, however, is far from a proof.

\subsection{Overview of the Proof of~\cref{thm:2CN_prime}}
Assume by way of contradiction that $\cP$ is not a prime $2$-CN. 
Thus, there exist 1-CNs $\cV_1,\ldots \cV_k$ such that $\lang(\cP)=\bigcap_{1 \leq j \leq k}\lang(\cV_j)$.

Throughout the proof, we focus on words of the form $a^{m_1} \# a^{m_2} \# \cdots \# a^{m_{k+1}} \# b^{m_b} c^{m_c}$ for integers $\left\{ m_i \right\}_{i=1}^{k+1}, m_b, m_c \in \bbN$. 
We index the $a^{m_i}$ segments of these words, so $a^{m_i}$ is the $i$-th segment.
Note that we focus on words with $k+1$ many $a$ segments, one more than the number of $\cV_j$ factors in the intersection. 
It is useful to think about each segment as ``paying'' for either $b$ or $c$. 
Then, a word is accepted if there is a way to choose for each segment whether it pays for $b$ or $c$, such that there is sufficient budget for both.

Let $i\in [k+1]$ and $j\in [k]$. 
We say that the $i$-th segment is \emph{bad} in $\cV_j$ if, 
 intuitively, we can pump the length $m_i$ of segment $i$ whilst pumping both $m_b$ and $m_c$ to unbounded lengths, such that the resulting words are accepted by $\cV_j$ (see~\cref{def:good_bad} for the formal definition).
%intuitively, we can arbitrarily increase the length $m_i$ of the $i$-th segment whilst increasing the lengths of $m_b$ and $m_c$, such that the resulting words are accepted by $\cV_j$ (formally described in~\cref{def:good_bad}).
For example, consider the word $a^{10} \# a^{10} \# a^{10} \# b^{20} c^{10} \in \lang(\cP)$. 
If the second segment is bad for $\cV_j$ then there exist $x, y, z > 0$ such that for every $t, t_b, t_c \in \bbN$ it holds that $a^{10} \# a^{10+tx} \# a^{10} \# b^{20+t_b y} c^{10+t_c z} \in \lang(\cV_j)$.
Observe that such behaviour is undesirable, since for large enough $t, t_b, t_c$, the resulting word is not in $\lang(\cP)$. 
Note, however, that the existence of such a bad segment is not a contradiction by itself, since the resulting pumped words might not be accepted by some other 1-CN $\cV_{j'}$. 

In order to reach a contradiction, we need to show the existence of a segment $i$ that is bad for \emph{every} $\cV_j$. 
Moreover, we must also show that arbitrarily increasing $m_i, m_b, m_c$ can be simultaneously achieved in all the $\cV_j$ together (i.e., the above $x, y, z > 0$ are the same for all $V_j$).
This would create a contradiction since all the $\cV_j$ accept a word that is not in $\lang(\cP)$.
Our goal is therefore to establish a robust and precise definition of a ``bad'' segment, then find a word $w$ comprising $k+1$ segments where one of the segments is bad for every $\cV_j$, and pumping the words in each segment can be done synchronously. 

\subsection{Pumping Arguments in One-Counter Nets}
\label{sec:1vass-properties}
In this section we establish some pumping results for 1-CN which will be used in the proof of~\cref{thm:2CN_prime}. 
Throughout this section, we consider a 1-CN $\cV=\tup{\Sigma,Q,Q_0,\delta,F}$.

Our first lemma states the intuitive fact that without $\positive$ cycles, the counter value of a run is bounded (proof in~\cref{apx:lem:maxcounternopositivecycles}).
\begin{lemma}\label{lem:maxcounternopositivecycles}
    Let $(q,n)$ be a configuration of $\cV$, let $W$ be the maximal positive update in $\cV$, $\sigma \in \Sigma$, and $N\in\bbN$. If an 
    $\bbN$-run $\rho$ of $\cV$ on $\sigma^N$ from configuration $(q,n)$ does not traverse any $\positive$ cycle, then the maximal possible counter value anywhere along $\rho$ is $n+W|Q|$. 
\end{lemma}
% \begin{proof}
    % We prove the contra-positive: assume $\rho$ has counter value $n+W|Q|$, w.l.o.g. we can assume this occurs at the end. Since the maximal counter increase along a transition is $W$, it follows that there are at least $|Q|$ indices where the counter increases beyond any previous level. Thus, we can find indices $i_0,\ldots,i_{|Q|}$ such that the counter increases in the infix of $\rho$ between $i_j$ and $i_{j+1}$. Since these are $|Q|+1$ states, then there exists a state that is visited twice, which forms a $\positive$ cycle in $\rho$.
% \hfill\qed\end{proof}

The next lemma shows that long-enough runs must contain $\nonnegative$ cycles.
\begin{lemma}\label{lem:largeenoughwordnonnegativecycle}
    %Consider a 1-CN $\cV=\tup{\Sigma,Q,Q_0,\delta,F}$. 
    Let $\sigma\in \Sigma$ and $(q,n)$ be an $\bbN$-configuration of $\cV$. Then there exists $N\in \bbN$ such that for all $N' \geq N$, every $\bbN$-run of $\cV$ on $\sigma^{N'}$ from $(q,n)$ traverses a $\nonnegative$ cycle. 
\end{lemma}
\begin{proof}
Let $W$ be the maximal positive transition update in $\cV$, we show that $N=|Q|(n+|Q| \cdot W)$ satisfies the requirements.
Assume by way of contradiction that $\cV$ can read $\sigma^{N}$ via an $\bbN$-run $\rho=(q_0,n_0=n)\stackrel{\rho}{\to} (q_N,n_N)$ that only traverses $\negative$ cycles. 

Since $\rho$ visits $N+1$ states, then by the Pigeonhole Principle, there exists a state $p \in Q$ that is visited $m \geq (N+1)/\abs{Q} > N/\abs{Q}$ many times in $\rho$.

Consider all the indices $0 \leq i_1 < i_2 < \ldots < i_m \leq N$ such that $p = q_{i_1} = \ldots = q_{i_N}$.	
Each run segment $(q_{i_1}, n_{i_1}) \rightarrow (q_{i_2}, n_{i_2}), \ldots, (q_{i_{m-1}}, n_{i_{m-1}}) \rightarrow (q_{i_m}, n_{i_m})$ is a cycle in $\rho$, and therefore must have negative effect.	
Thus $n_{i_1} > n_{i_2} > \ldots > n_{i_m} \geq 0$, so in particular $n_{i_1} \geq n_{i_m} + m \geq 0$ (as each cycle has effect at most $-1$).	
Moreover, $n_{i_1} < n + \abs{Q}\cdot W$ since the prefix $(q_0, n) \rightarrow (q_{i_1}, n_{i_1})$ cannot contain a non-negative cycle.	
However, since $m >N/\abs{Q} = n + \abs{Q}\cdot W$ and $n_{i_1} \geq n_{i_m} + m \geq m \geq n + \abs{Q}\cdot W$, we get $n+|Q|\cdot W<n+ |Q| \cdot W$ which is a contradiction.
% Consider all the indices $0 \leq r < s < \ldots < i_m \leq N$ such that $p = q_{r} = \ldots = q_{i_N}$.
% Each run segment $(q_{r}, n_{r}) \rightarrow (q_{s}, n_{s}), \ldots, (q_{i_{m-1}}, n_{i_{m-1}}) \rightarrow (q_{i_m}, n_{i_m})$ is a cycle in $\rho$, and therefore must have negative effect.
% Thus $n_{r} > n_{s} > \ldots > n_{i_m} \geq 0$, so in particular $n_{r} \geq n_{i_m} + m \geq 0$, as each cycle has effect at most $-1$.

%Moreover, $n_{r} < n + \abs{Q}\cdot W$ since the prefix $(q_0, n) \rightarrow (q_{r}, n_{r})$ cannot contain a non-negative cycle.
%However, since $m >N/\abs{Q} = n + \abs{Q}\cdot W$ and $n_{r} \geq n_{i_m} + m \geq m \geq n + \abs{Q}\cdot W$, we get $n+|Q|\cdot W<n+ |Q|W$ which is a contradiction.
\hfill\qed\end{proof}

Next, we show that runs with $\nonnegative$ and $\positive$ cycles have ``pumpable'' infixes.
\begin{lemma}
\label{lem:cyclethensimplecycle}
Let $\sigma\in \Sigma$ and consider a $\positive$ (resp. $\nonnegative$) cycle $\pi=(q_0,c_0)\stackrel{\sigma}{\to} (q_1,c_1) \stackrel{\sigma}{\to} \ldots (q_n=q_0,c_n)$ on $\sigma^n$ that induces an $\bbN$-run. Then, there is a sequence of (not necessarily contiguous) indices $0 \leq i_1 \leq \ldots \leq i_k \leq n$ such that $q_{i_1}\stackrel{\sigma}{\to} q_{i_2}\stackrel{\sigma}{\to}\cdots q_{i_k}$ is a simple $\positive$ (resp. $\nonnegative$) cycle with some effect $e>0$ (resp. $e \geq 0$).	
In addition, this simple cycle is ``pumpable'' from the first occurrence of $q_{i_1}$ in $\pi$; namely, for all $m \in \bbN$ there is a run $\pi_m$ obtained from $\pi$ by traversing the cycle $m$ times so that	
$\effect{\pi_m}=\effect{\pi}+em$.	
\end{lemma}

% \begin{lemma}\label{lem:cyclethensimplecycle}
%     Let $\sigma\in \Sigma$ and consider a $\positive$ (resp. $\nonnegative$) cycle $\pi=(q_0,c_0)\stackrel{\sigma}{\to} (q_1,c_1) \stackrel{\sigma}{\to} \cdots \stackrel{\sigma}{\to} (q_n=q_0,c_n)$ on $\sigma^n$ that induces an $\bbN$-run. 
%     Then, there is a sequence of (not necessarily contiguous) indices $0 \leq r \leq \ldots \leq i_k \leq n$ such that $q_{r}\stackrel{\sigma}{\to} q_{s} \stackrel{\sigma}{\to} \cdots \stackrel{\sigma}{\to} q_{i_k}$ is a simple $\positive$ or $\nonnegative$ cycle with some effect $e>0$ or $e \geq 0$, respectively. 
    
%     In addition, this simple cycle is \emph{pumpable} from the first occurrence of $q_{r}$ in $\pi$; for all $m \in \bbN$ there is a run $\pi_m$ obtained from $\pi$ by traversing the cycle $m$ times so that
%     $\effect{\pi_m}=\effect{\pi}+em$. 
% \end{lemma}
\begin{proof}
    We prove the $\nonnegative$ case, the $\positive$ case can be proved mutatis mutandis. 
    
    We define $\pi_m=(q_0,c_0)\stackrel{\sigma}{\to} \ldots (q_{i_1},c_{i_1}) \stackrel{\sigma}{\to} \ldots (q_{i_1},c_{i_1}+km)\stackrel{\sigma}{\to} \ldots (q_n,c_n+km)$. The proof is now by induction on the length of $\pi$.
    
    The base of the induction is a cyclic $\bbN$-run of length $2$. In this case $\pi=(q_0,c_0) \stackrel{\sigma}{\to} (q_1=q_0, c_1)$ is itself a $\nonnegative$ simple cycle that is infinitely pumpable from $(q_0,c_0)$.

    	We now assume correctness for length $n$, and discuss $\pi=(q_0,c_0)\stackrel{\sigma}{\to} (q_1,c_1) \stackrel{\sigma}{\to} \ldots (q_n=q_0,c_n)$ of length $n+1$. Let $0\leq j_1<j_2\leq n$ be indices such that $q_{j_1}=q_{j_2}$, for a maximal $j_1$. Note that the cycle $\tau=(q_{j_1},c_{j_1})\stackrel{\sigma}{\to} \ldots (q_{j_2},c_{j_2})$ must be simple. 
        If $j_1=0$ and $j_2=n$, then $\pi$ itself is a simple $\nonnegative$ cycle, and the pumping argument is straightforward. 
     Otherwise $\tau$ is nested.
    We now split into two cases, based on whether $\effect{\tau} \geq 0$.
    \begin{enumerate}
        \item $\tau$ is $\nonnegative$: then the induction hypothesis applies on $\tau$. 
        We take the guaranteed constants $j_1 \leq i_1 \leq \ldots \leq i_k \leq j_2$, which apply to $\pi$ as well.
        \item $\tau$ is $\negative$: then we remove $\tau$ from $\pi$ to obtain $\pi'=(q_0,c_0)\stackrel{\sigma}{\to} \ldots (q_{j_1},c_{j_1}) \stackrel{\sigma}{\to} (q_{j_2+1},c'_{j_2+1}) \stackrel{\sigma}{\to} \ldots (q_n,c'_n)$, such that $c'_i \geq c_i$ for all $j_2+1 \leq i \leq n$. The induction hypothesis applies on $\pi'$, so let $i_1, \ldots, i_k$ be the guaranteed constants. Note that $i_1 \leq j_1$, since the cycle removed when obtaining $\pi'$ from $\pi$ is the last occurrence of a repetition of states in $\pi$. We therefore know that $q_{i_1}\stackrel{\sigma}{\to} q_{i_2}\stackrel{\sigma}{\to}\cdots q_{i_k}$ is a simple $\nonnegative$ cycle in $\pi'$ -- which applies to $\pi$ as well. In addition, it is infinitely pumpable from $\bbN$-configuration $(q_{i_1},c_{i_1})$ in $\pi'$ for $i_1 \leq j_1$. Indeed, since $\pi$ and $\pi'$ coincide up to and including $(q_{j_1},c_{j_1})$ between $\pi$ and $\pi'$ - this cycle is infinitely pumpable in $\pi$ as well.
        \qed
    \end{enumerate}
\end{proof}

The simple cycle in~\cref{lem:cyclethensimplecycle} has length $k<|Q|$. By pumping it $\frac{|Q|!}{k}$ times we obtain a pumpable cycle of length $|Q|!$, allowing us to conclude with the following. 
\begin{corollary} \label{obs:nonnegativepumptofactorial}
    Let $\rho$ be an $\bbN$-run of $\cV$ on $\sigma^n$ that traverses a $\nonnegative$ cycle. For every $m\in \bbN$, we can construct an $\bbN$-run $\rho'$ of $\cV$ on $\sigma^{n+m|Q|!}$ such that $\effect{\rho'} \geq \effect{\rho}$ by pumping a $\nonnegative$ simple cycle in $\rho$.
\end{corollary}

\subsection{Good and Bad Segments}
\label{sec:bad_segments}
We lift the colour scheme\footnote{The colours were chosen as accessible for the colourblind. For a greyscale-friendly version, see \cref{apx:grey}.} of $\positive$ and $\nonnegative$ to words and runs as follows. 
For a word $w=uv$ and a run $\rho$, we write e.g., $\posCol{u}\nonnegCol{v}$ to denote that $\rho$ traverses a $\positive$ cycle when reading $u$, then a $\nonnegative$ cycle when reading $v$. 
Note that this does not preclude other cycles, e.g., there could also be negative cycles in the $u$ part, etc. That is, the colouring is not unique, but represents elements of the run.

Recall our assumption that $\lang(\cP) = \bigcap_{1 \leq j \leq k}\lang(\cV_j)$, and denote $\cV_j = \tup{ \Sigma, Q_j, I_j, \delta_j, F_j}$ for all $j \in [k]$. 
Let $Q_{\max} = \max\set{|Q_j|}_{j=1}^k$ and denote $\alpha=Q_{\max}!$. 
Further recall that we focus on words of the form $a^{m_1} \# a^{m_2} \# \cdots \# a^{m_{k+1}} \# b^{m_b} c^{m_c}$ for integers $\left\{ m_i \right\}_{i=1}^{k+1}, m_b, m_c \in \bbN$, and that we refer to the infix $a^{m_i}$ as the $i$-th segment, for $1 \leq i \leq k$.
We proceed to formally define good and bad segments.
\begin{definition}[Good and Bad Segments]
\label{def:good_bad}
The $i$-th segment is \emph{bad} in $\cV_j$ if there exist constants $\left\{m_i\right\}_{i=1}^{k+1},m_b,m_c \in \bbN$ such that the following hold.
\begin{enumerate}[label=(\alph*)]
    \item $\left\{ m_i \right\}_{i=1}^{k+1}, m_b,m_c$ are multiples of $\alpha$, and
    \item there is an accepting $\bbN$-run $\rho$ of $\cV_j$ on $w = a^{m_1} \# a^{m_2} \# \cdots \# a^{m_{k+1}} \# b^{m_b} c^{m_c}$ that adheres to one of the three forms:
    \begin{enumerate}[label=(\roman*)]
        \item $\nonnegCol{a^{m_1}}\#\nonnegCol{a^{m_2}}\#\cdots \nonnegCol{a^{m_{i-1}}}\#\posCol{a^{m_i}}\#a^{m_{i+1}}\#\cdots\#a^{m_{k+1}}\#b^{m_b}c^{m_c}$,
        \item $\nonnegCol{a^{m_1}}\#\nonnegCol{a^{m_2}}\#\cdots \nonnegCol{a^{m_{i-1}}}\#\nonnegCol{a^{m_i}}\#\nonnegCol{a^{m_{i+1}}}\#\cdots\#\nonnegCol{a^{m_{k+1}}}\#\nonnegCol{b^{m_b}}\nonnegCol{c^{m_c}}$, or
        \item $\nonnegCol{a^{m_1}}\#\nonnegCol{a^{m_2}}\#\cdots \nonnegCol{a^{m_{i-1}}}\#\nonnegCol{a^{m_i}}\#\nonnegCol{a^{m_{i+1}}}\#\cdots\#\nonnegCol{a^{m_{k+1}}}\#\posCol{b^{m_b}}c^{m_c}$.
    \end{enumerate}
\end{enumerate}
The $i$-th segment is \emph{good} in $\cV_j$ if it is not bad in $\cV_j$.
\end{definition}

\cref{lem:bad_segements_are_bad} formalises the intuition that a bad segment can be pumped simultaneously with both the $b$ and $c$ segments, giving rise to a word accepted by $\cV_j$ but rejected by $\cP$.

Intuitively, Forms (ii) and (iii) indicate that all segments are bad. 
Indeed, the $i$-th segment has a $\nonnegative$ cycle, so it can be pumped safely, and in Form (ii) both $b$ and $c$ can be pumped using $\nonnegative$ cycles.
Whereas in Form (iii) we can pump $b$ using a $\positive$ cycle, and can use it to compensate for pumping $c$, even if the latter requires iterating a negative cycle.

Form (i) is the interesting case, where we use a $\positive$ cycle in the $i$-th segment to compensate for pumping both $b$ and $c$. 
The requirement that all segments up to the $i$-th are $\nonnegative$ is at the core of our proof and is explained in~\cref{sec:proof_of_2Prime}.

\begin{lemma}
    \label{lem:bad_segements_are_bad}
    Suppose the $l$-th segment is bad in $\cV_j$, then there exist $x, y, z \in \bbN$, that are  multiples of $\alpha$, such that for every $n \in \bbN$ the following word $w$ is accepted by $\cV_j$.
    \begin{equation*}
        w_n = a^{m_1} \# a^{m_2} \# \cdots \# a^{m_{l-1}} \# a^{m_l+xn} \# a^{m_{l+1}} \# \cdots \# a^{m_{k+1}} \# b^{m_b+yn} c^{m_c+zn}
    \end{equation*}
\end{lemma}
\begin{proof}
    We can choose $z=\alpha$, then take $y$ to be large enough so that Form (iii) runs can compensate for negative cycles in $c^z$ using $\positive$ cycles in $b^y$, whilst not decreasing the counters in Form (ii) runs.
    We can indeed find such a $y \in \bbN$ that is a multiple of $\alpha$, since $\alpha$ is divisible by all lengths of simple cycles. 
    Finally, we choose $x$ so that Form (i) runs can compensate for $c^z$ and $b^y$ using $\positive$ cycles on $a^x$ in the $l$-th segment, again whilst not decreasing the counters in Forms (ii) and (iii).
\hfill\qed\end{proof}

Recall that our goal is to show that there is a segment $l \in [k+1]$ that is bad in \emph{every} $V_j$, for $j \in [k]$. 
In~\cref{lem:two_segments_one_bad}, We show that each $V_j$ has at most one good segment. 
Therefore, there are at most $k$ good segments in total, leaving at least one segment that is bad in every $\cV_j$, as desired.

\begin{lemma}
\label{lem:two_segments_one_bad}
    Let $j \in [k]$ and $0 \leq r < s \leq k+1$.
    Then the $r$-th or $s$-th segment is bad in $\cV_j$.
\end{lemma}
\begin{proof}
    Since $j$ is fixed, denote $\cV_j=\tup{\Sigma,Q,Q_0,\delta,F}$. 
    We inductively define constants $\left\{n_i\right\}_{i=1}^{k+1},n_b,n_c \in \bbN$ as follows. 
    Suppose that $n_1$ is a large-enough multiple of $\alpha$ so that~\cref{lem:largeenoughwordnonnegativecycle} guarantees a $\nonnegative$ cycle in any accepting run of $\cV_j$ on $a^{n_1}$ from some $(q_0,0)$ with $q_0\in Q_0$. 
    Now, assume that we have defined $n_1, \ldots n_{l-1}$, and consider the word $u = a^{n_1} \# a^{n_2} \# \cdots \# a^{n_{l-1}} \#$. 
    Define $n=|u| \cdot W$ where $W$ is the maximal update of any transition of $\cV_j$.
    Since $u$ consists of $\frac{n}{W}$ letters, $n+1$ is greater than any counter value that can be observed in any run of $\cV_j$ on $u$. 
    We define $n_l$ to be a multiple of $\alpha$ large enough so that~\cref{lem:largeenoughwordnonnegativecycle} guarantees a $\nonnegative$ cycle when reading $a^{n_l}$ from any configuration of the form $\{(q,n') \mid q \in Q,\ n' \leq n+1\}$.
    We set $n_b=n_c=\alpha$, the choice of $n_b,n_c$ is somewhat arbitrary. 
    Finally, we set $w = a^{n_1} \# \cdots \# a^{n_{k+1}} \# b^{n_b} c^{n_c}$.

    Now, for every $x\in \bbN$, we obtain from $w$ a word $w_x$ by pumping $x\alpha$ many $a$'s in the $r$-th and $s$-th segments and pumping $x\alpha$ many $b$'s and $c$'s in their segments.
    That is, let $n'_i = n_i + x\alpha$ for $i \in \set{r, s}$ and $n'_i=n_i$ for $i\notin \{r,s\}$, and let $n'_b=n_b+x \alpha$ and $n'_c=n_c+x \alpha$, then $w_x = a^{n'_1} \# \cdots \# a^{n'_{k+1}} \# b^{n'_b} c^{n'_c}$.
    Observe that $w_x\in \lang(\cP)$. 
    Indeed, since $n_{r} \geq n_b = \alpha$ and $n_{s} \geq n_c=\alpha$ we have that $n_{r}+x\alpha \geq n_b+x\alpha$ and $n_{s}+x\alpha \geq n_c+x\alpha$, so the $r$-th and $s$-th segments can already pay for the $b$'s and $c$'s, respectively.
    In particular, $w_x\in \lang(\cV_j)$ via some accepting $\bbN$-run $\rho_x$.

    We choose a particular value of $x$, as follows. 
    Consider $x$ and suppose some accepting $\bbN$-run $\rho_x$ as above does not traverse a $\positive$ cycle neither in $r$-th nor $s$-th segment.
    By~\cref{lem:maxcounternopositivecycles}, the maximal possible counter value of $\rho_x$ after reading
    \begin{equation*}
        a^{n_1} \# \cdots\# a^{n_{r}+x\alpha} \# \cdots \# a^{n_{s}+x\alpha} \# \cdots \# a^{n_{k+1}} \#
    \end{equation*} 
    is $M_b = (k+1+\sum_{z \in [k+1]\setminus\set{r,s}}n_z) \cdot W + 2|Q| \cdot W$. 
    Crucially, this value does not depend on $x$. 
    Further, if there is no $\positive$ cycle in the segment of $b$'s as well, again the maximal counter value of $\rho$ up to the $c$ segment is bounded by $M_c = (k+2 + \sum_{z\in [k+1]\setminus\set{r,s}}n_z) \cdot W + 3|Q| \cdot W$, that is independent of $x$ and $M_b$. 
    By~\cref{lem:largeenoughwordnonnegativecycle}, we can now choose $x$ large enough to satisfy that for every accepting $\bbN$-run $\rho_x$ on $w_x$:
    \begin{enumerate}[label=\arabic*.]
        \item If $\rho_x$ does not traverse any $\positive$ cycle in the $r$-th or $s$-th segments, then $\rho_x$ has a $\nonnegative$ cycle reading $b^{(n_b+x\alpha)}$ from any configuration in $\set{ (q,M') \mid q\in Q,\ M' \leq M_b }$.
        \item If $\rho_x$ does not traverse any $\positive$ cycle in the $r$-th or $s$-th segment, nor in the $b$ segment, then $\rho_x$ has a $\nonnegative$ cycle reading $c^{(n_c+x\alpha)}$ from any configuration in $\set{ (q,M') | q\in Q,\ M'\le M_c }$. 
    \end{enumerate}
    
    Having fixed $x$, we claim that for the constants of $w_x$, one of the $r$-th or $s$-th segment is bad in $\cV_j$.
    By construction,~\cref{lem:largeenoughwordnonnegativecycle} guarantees that $\rho_x$ has $\nonnegative$ cycles in segments $1,\ldots r-1$. 
    If $\rho_x$ has a $\positive$ cycle in segment $r$, then $\rho_x$ is of Form (i):
    \begin{equation*}
        \nonnegCol{a^{n_1}} \# \nonnegCol{a^{n_2}} \# \cdots \# \nonnegCol{a^{n_{r-1}}} \# \posCol{a^{n_{r}+x\alpha}} \# \cdots \# a^{n_{s}+x\alpha} \# \cdots \#  a^{n_{k+1}} \# b^{n_b+x\alpha} c^{n_c+x\alpha}
    \end{equation*}
    and so the $r$-th segment must be bad in $\cV_j$.

    Otherwise, if $\rho_x$ does not have a $\positive$ cycle in the $r$-th segment, then the construction in~\cref{lem:largeenoughwordnonnegativecycle} guarantees $\nonnegative$ cycles in segments indexed $r, r+1, \ldots, s-1$. 
    Indeed, for the $r$-th segment, we are guaranteed a $\nonnegative$ cycle reading $a^{n_{r}}$, all the more for $a^{n_{r}+x\alpha}$.
    As for segments indexed $r+1, \ldots s-1$, if $\rho_x$ does not have a $\positive$ cycle in the $r$-th segment, then the maximal effect of segment $r$ is $|Q| \cdot W$. 
    However, $n_{r+1}$ was constructed to guarantee a $\nonnegative$ cycle even in case the effect of segment $r$ is $Wn_{r} \geq W\alpha \geq W|Q|$. 
    
    If there is a $\positive$ cycle in segment $s$, then $\rho_x$ is again of Form (i): 
    \begin{equation*}
        \nonnegCol{a^{n_1}} \# \nonnegCol{a^{n_2}} \# \cdots \# \nonnegCol{a^{n_{s-1}}} \# \posCol{a^{n_{s}+x\alpha}} \# a^{n_{s+1}} \# \cdots \# a^{n_{k+1}} \# b^{n_b+x\alpha} c^{m_c+x\alpha}
    \end{equation*}
    and so the $s$-th segment must be bad in $\cV_j$.

    Otherwise, using the same arguments as for the $r$-th segment, we have that segments indexed $s+1,\ldots,i_{k+1}$ each contain a $\nonnegative$ cycle. 
    In this case we are left with the $b$ and $c$ segments. 
    The choice of $x$ guarantees a $\nonnegative$ cycle in the $b$ segment.
    If $\rho_x$ traverses a $\positive$ cycle in the $b$ segment, then $w_x$ is of Form (iii).
    \begin{equation*}
        \nonnegCol{a^{n_1}} \# \nonnegCol{a^{n_2}} \# \cdots \# \nonnegCol{a^{n_{k+1}}} \# \posCol{b^{n_b+x\alpha}} c^{n_c+x\alpha}
    \end{equation*}
    Finally, if there are no $\positive$ cycles in the $b$ segment, then the choice of $x$ again guarantees a $\nonnegative$ cycle in the $c$ segment, so $w_x$ is of Form (ii). 
    \begin{equation*}
    \nonnegCol{a^{n_1}} \# \nonnegCol{a^{n_2}} \# \cdots \# \nonnegCol{a^{n_{k+1}}} \# \nonnegCol{b^{n_b+x\alpha}} \nonnegCol{c^{n_c+x\alpha}}
    \end{equation*}
    In the two latter cases, both the $r$-th and the $s$-th segments are bad in $\cV_j$.
\hfill\qed\end{proof}

\subsection{Proof of~\cref{thm:2CN_prime}}
\label{sec:proof_of_2Prime}
Given~\cref{lem:two_segments_one_bad}, we now know that each $\cV_j$ has at most one good segment.
Therefore, all $1$-CNs $\cV_1,\ldots,\cV_k$ together have at most $k$ good segments. 
Recall that the words we focus on have $k+1$ segments, and therefore there is at least one segment, say the $l$-th segment, that is bad in every $\cV_j$. 
Note, however, that this segment may correspond to different constants in each $\cV_j$. 
That is, there exists constants $\set{ m_i^j, m_b^j, m_c^j \mid i \in [k+1], j \in [k] }$ witnessing that the $l$-th segment is bad for each $\cV_j$. 
We group the $\cV_j$ according to the form of their accepting runs $\rho_j$ (see Definition~\ref{def:good_bad}):
\begin{enumerate}[label=(\roman*)]
    \item $\nonnegCol{a^{m^j_1}} \# \nonnegCol{a^{m^j_2}} \# \cdots \# \posCol{a^{m^j_l}} \# a^{m^j_{l+1}} \# \cdots \# a^{m^j_{k+1}} \# b^{m^j_b}c^{m^j_c}$, 
    \item $\nonnegCol{a^{m^j_1}} \# \nonnegCol{a^{m^j_2}} \# \cdots \# \nonnegCol{a^{m^j_l}} \# \nonnegCol{a^{m^j_{l+1}}} \# \cdots \# \nonnegCol{a^{m^j_{k+1}}} \# \nonnegCol{b^{m^j_b}} \nonnegCol{c^{m^j_c}}$, or
    \item  $\nonnegCol{a^{m^j_1}} \# \nonnegCol{a^{m^j_2}} \# \cdots \# \nonnegCol{a^{m^j_l}} \# \nonnegCol{a^{m^j_{l+1}}} \# \cdots \# \nonnegCol{a^{m^j_{k+1}}} \# \posCol{b^{m^j_b}}c^{m^j_c}$. 
\end{enumerate}

{
    \setlength{\abovedisplayskip}{0pt}
    \setlength{\belowdisplayskip}{0pt}
    We now find constants resulting in a single word for which the $l$-th segment is bad in every $\cV_j$.
    First, for $i \in [k+1] \setminus\set{l}$, we define $M_i = \max\set{m^j_i \mid j \in [k]}$, note that these values are still multiples of $\alpha$. 
    Similarly, we define $M_c = \max\set{m^j_c \mid j \in [k]}$. 
    It remains to fix new constants $\newl$ and $\newb$, which we do in phases in the following.
    The resulting word is then
    \begin{equation*}
        w = a^{M_1} \# \cdots \# a^{M_{l-1}} \# a^{{\newl}} \# a^{M_{l+1}} \# \cdots \# a^{M_{k+1}} \# b^{{\newb}} c^{M_c}.    
    \end{equation*}
    
    Most steps in the analysis below are based on~\cref{lem:cyclethensimplecycle,obs:nonnegativepumptofactorial}. 
    We first, partially, handle Form (iii) runs. 
    For such $\cV_j$, there is an accepting $\bbN$-run $\rho_j$ on 
    \begin{equation*}
        \nonnegCol{a^{m^j_1}} \# \cdots \# \nonnegCol{a^{m^j_{l-1}}} \# \nonnegCol{a^{m^j_l}} \# \nonnegCol{a^{m^j_{l+1}}} \# \cdots \# \nonnegCol{a^{m^j_{k+1}}} \# \posCol{b^{m^j_b}} c^{m^j_c}
    \end{equation*}
    By pumping $\nonnegative$ cycles as per~\cref{obs:nonnegativepumptofactorial} in all segments except $l$ we obtain an accepting $\bbN$-run $\rho'_j$ on 
    \begin{equation*}
        \nonnegCol{a^{M_1}} \# \cdots \# \nonnegCol{a^{M_{l-1}}} \# \nonnegCol{a^{m^j_l}} \# \nonnegCol{a^{M_{l+1}}} \# \cdots \# \nonnegCol{a^{M_{k+1}}} \# \posCol{b^{m^j_b}} c^{m^j_c}.
    \end{equation*}
    We now pump arbitrary cycles in the $c$ segment to construct a $\bbZ$-run $\rho''_j$ on \begin{equation*}
        \nonnegCol{a^{M_1}} \# \cdots \# \nonnegCol{a^{M_{l-1}}} \# \nonnegCol{a^{m^j_l}} \# \nonnegCol{a^{M_{l+1}}} \# \cdots \# \nonnegCol{a^{M_{k+1}}} \# \posCol{b^{m^j_b}} c^{M_c}.
    \end{equation*} 
    Next, we compensate for possible negative cycles in the $c$ segment by pumping a $\positive$ cycle in the $b$ segment. 
    Thus, we construct an $\bbN$-run $\rho'''_j$ on 
    \begin{equation*}
        \nonnegCol{a^{M_1}} \# \cdots \# \nonnegCol{a^{M_{l-1}}} \# \nonnegCol{a^{m^j_l}} \# \nonnegCol{a^{M_{l+1}}} \# \cdots \# \nonnegCol{a^{M_{k+1}}} \# \posCol{b^{\newb}} c^{M_c},
    \end{equation*} 
    where $\newb$ is chosen to be large enough such that $\rho'''_j$ is an $\bbN$-run for all $\cV_j$, $j \in [k]$.
    Note that it remains to fix $\newl$.
    
   We now turn to Form (i) with a similar process
    We start with an accepting $\bbN$-run $\rho_j$ on 
    \begin{equation*}
        \nonnegCol{a^{m^j_1}} \# \cdots \# \nonnegCol{a^{m^j_{l-1}}} \# \posCol{a^{m^j_l}} \# a^{m^j_{l+1}} \# \cdots \# a^{m^j_{k+1}} \# b^{m^j_b} c^{m^j_c}.
    \end{equation*}
    Pump $\nonnegative$ cycles in segments indexed $1, \ldots, l-1$ to obtain an accepting $\bbN$-run $\rho'_j$ on  
    \begin{equation*}
        \nonnegCol{a^{M_1}} \# \cdots \# \nonnegCol{a^{M_{l-1}}} \# \posCol{a^{m^j_l}} \# a^{m^j_{l+1}} \# \cdots \# a^{m^j_{k+1}} \# b^{m^j_b}c^{m^j_c}.
    \end{equation*}
    Now, obtain a $\bbZ$-run $\rho''_j$ by pumping arbitrary cycles in the remaining segments, including the $b$ segment.
    \begin{equation*}
        \nonnegCol{a^{M_1}} \# \cdots \# \nonnegCol{a^{M_{l-1}}} \# \posCol{a^{m^j_l}} \# a^{M_{l+1}} \# \cdots \# a^{M_{k+1}} \# b^{\newb} c^{M_c}
    \end{equation*}
    Again, compensate for negative cycles by taking $\newl$ large enough so that pumping $\positive$ cycles in the $l$-th segment yields an accepting $\bbN$-run $\rho'''_j$ on
    \begin{equation*}
        \nonnegCol{a^{M_1}} \# \cdots \# \nonnegCol{a^{M_{l-1}}} \# \posCol{a^{\newl}} \# a^{M_{l+1}} \# \cdots \# a^{M_{k+1}} \# b^{\newb} c^{M_c}.
    \end{equation*}
    
    We now return to Form (iii) and fix the $l$-th segment by pumping $\nonnegative$ cycles to construct an accepting $\bbN$-run on 
    \begin{equation*}
        \nonnegCol{a^{M_1}} \# \cdots \# \nonnegCol{a^{M_{l-1}}} \# \nonnegCol{a^{\newl}} \# \nonnegCol{a^{M_{l+1}}} \# \cdots \# \nonnegCol{a^{M_{k+1}}} \# \posCol{b^{\newb}} c^{M_c}.
    \end{equation*}
    
    We are left with Form (ii), which are the most straightforward to handle. 
    We simply pump $\nonnegative$ cycles in all segments to construct an accepting $\bbN$-run $\rho'_j$ on
    \begin{equation*}
        \nonnegCol{a^{M_1}} \# \cdots \# \nonnegCol{a^{M_{l-1}}} \# \nonnegCol{a^{\newl}} \# \nonnegCol{a^{M_{l+1}}} \# \cdots \# \nonnegCol{a^{M_{k+1}}} \# \nonnegCol{b^{\newb}} \nonnegCol{c^{M_c}}.
    \end{equation*}
}
Note that the requirement for all segments before the $l$-th to be $\nonnegative$ is crucial, otherwise we won't be able to pump all the cycles in all forms simultaneously.

We now have that $w$ is accepted by every $\cV_j$, and the $l$-th segment is bad for all $\cV_j$.
By applying~\cref{lem:bad_segements_are_bad} for each of the $\cV_j$ and taking global constants to be the products of the respective constants $x, y, z > 0$ for each $\cV_j$, we now obtain $X, Y, Z \in \bbN$, multiples of $\alpha$, such that for every $n\in \bbN$ the word 
\begin{equation*}
    w_n = a^{M_1} \# \cdots \# a^{M_{l-1}} \# a^{{\newl+Xn}} \# a^{M_{l+1}} \# \cdots \# a^{M_{k+1}} \# b^{{\newb+Yn}} c^{M_c+Zn} \in \lang(\cV_j) \;\; \forall j \in [k].   
\end{equation*}
is accepted by every $\cV_j$.

Finally, we choose $n$ large enough to satisfy $\sum_{i\in [k+1] \setminus\set{l}} M_i < \min\{\newb+Yn, M_c+Zn\}$, so that $w_n\notin \lang(\cP)$.
This is possible because, w.l.o.g, the $l$-th segment can only pay for $b$, and the remaining segments $[k+1] \setminus\set{l}$ cannot pay for $c$.
This contradicts the assumption that $\lang(\cP) = \bigcap_{j\in [k]}\lang(\cV_j)$, concluding the proof of~\cref{thm:2CN_prime}.
\hfill\qed

\begin{remark}[Unbounded Compositeness]
\label{rmk:unbounded_compositeness}
The proof of~\cref{thm:2CN_prime} shows that if words with $k+1$ segments are allowed, then the language is not $(1,k)$-composite, we use this to establish primality. 
By intersecting $\lang(\cP)$ with words that allow at most $k+1$ segments, we obtain a language that is not $(1,k)$-composite, but it is not hard to show that it is $(1,2^{k+1})$-composite. 
This demonstrates that a $2$-CN can be composite, but may require unboundedly many factors.
\end{remark}

The intuition behind~\cref{thm:2CN_prime} is that separate counters are needed to keep track of the elements that ``cover'' $b^{m_b}$ and $c^{m_c}$. Extending this idea to $k$-CN, we require that the $a$ segments are partitioned to $k$ different sets that cover $k$ ``targets''.
\begin{conjecture}
\label{conj:prime kCN}
The following language is the language of a prime $k$-CN:
\begin{align*}
  L_k=&\{a^{m_1}\# a^{m_2}\# \cdots \#a^{m_t}\#b_1^{n_1}\#b_2^{n_2}\cdots \# b_k^{n_k} \mid\\
 &\exists I_1,\ldots,I_k\subseteq [t]\ \forall i\in [k],\ \sum_{j\in I_i}m_j\ge n_i\ \wedge \forall i\neq j,\ I_i\cap I_j=\emptyset \}  
\end{align*}
\end{conjecture}
While constructing a $k$-CN for $L_k$ is a simple extension of~\cref{xmp:prime2CN}, proving that it is indeed prime does not seem to succumb to our techniques, and we leave it as an important open problem (see~\cref{sec:discussion}).

\section{Primality of Counter Nets is Undecidable}\label{sec:primality_undecidable}
In this section we consider the \emph{primality} and \emph{dimension-minimality} decision problems: given a $k$-CN $\cA$, decide whether $\cA$ is prime and whether $\cA$ is dimension-minimal, respectively.

We use our prime $2$-CN from~\cref{xmp:prime2CN} and the results of~\cref{sec:prime-2vass} to show that both problems are undecidable. 
	Our proof is by reduction from the containment problem\footnote{Actually, the complement thereof.} for $1$-CN: given two $1$-CN $\cA,\cB$ over alphabet $\Sigma$, decide whether $\lang(\cA)\subseteq \lang(\cB)$. This problem was shown to be undecidable in~\cite{hofman2013decidability}.

We begin by describing the reduction that applies to both problems.
Consider an instance of 1-CN containment with two 1-CNs $\cA$ and $\cB$ over the alphabet $\Sigma$.
We construct a 2-CN $\cC$ as follows. 
Let $\Lambda$ be the alphabet of the $2$-CN from~\cref{xmp:prime2CN} and Theorem~\ref{thm:2CN_prime}, and let $\$ \notin \Sigma \cup \Lambda$ be a fresh symbol.
Intuitively, $\cC$ accepts words of the form $u\$v$ when either $u\in \lang(\cA)$ and $v$ is accepted by $\cP$ starting from the maximal counter $\cA$ ended with on $u$, or when $u\in \lang(\cB)$ and $v\in \Lambda^*$.

%Intuitively, we construct $\cC$ so that its language contains words of the form $u\$v$, here  either $u \in \lang(\cB)$ and $v\in \Lambda^*$, or $u\in \lang(\cA)$ and $v$ is accepted by $\cP$ when starting from the maximal counter $\cA$ ended with on $u$. 
%Note that for $u \in \lang(\cB)$ and any word $v \in \Lambda^*$, we have $u\$ v \in \lang(\cC)$. 

Formally, we convert $\cA$ and $\cB$ to 2-CNs $\cA'$ and $\cB'$ by adding a counter and never modifying its value, so a transition $(p, \sigma, v, q)$ in $\cA$ becomes $(p, \sigma, (v,0), q))$ in $\cA'$, for example. 
We construct a 2-CN $\cC$ as follows (see~\cref{fig:primalityReduction}). 
We take $\cA'$, $\cB'$, and $\cP$, and for every accepting state $q$ of $\cA'$ we introduce a transition $(q, \$, \vec{0}, p_0)$ where $p_0$ is an initial state of $\cP$. 
We then add a new accepting state $q_\top$ and add the transitions $(q_\top, \lambda, \vec{0}, q_\top)$ for every letter $\lambda \in \Lambda$, in other words $q_\top$ is an accepting sink for $\Lambda$.
We also add transitions $(s, \$, \vec{0}, q_\top)$ from every accepting state $s$ of $\cB'$.
The initial states are those of $\cA'$ and $\cB'$, and the accepting states are those of $\cP$ and $q_\top$. 

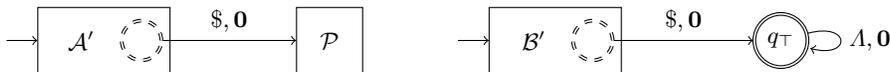
\begin{figure}[ht]
\begin{tikzpicture}[->, node distance=2cm] 
\node[state,initial, label={[label distance=-25pt]left:$\cA'$}, initial text={}, rectangle, minimum height=25pt, minimum width=50pt](A) at (0,0) {};
\node[state, rectangle, minimum height=25pt, minimum width=25pt] (B) at (3,0) {$\cP$};
\node[state, initial, label={[label distance=-25pt]left:$\cB'$}, initial text={}, rectangle, minimum height=25pt, minimum width=50](C) at (6,0) {};
\node[state, accepting,inner sep=3pt, minimum size=3pt](D) at (9,0) {$q_\top$};
\node[state, accepting, dashed, minimum size=15pt](E) at (0.5,0) {};
\node[state, accepting, dashed, minimum size=15pt](F) at (6.5,0) {};

\draw
    (E) edge[above] node{$\$, \vec{0}$} (B)
    (F) edge[above] node{$\$, \vec{0}$} (D)
    (D) edge[loop right] node{$\Lambda, \vec{0}$} (D);

\end{tikzpicture}
\caption{The reduction from $1$-CN non-containment to $2$-CN primality and dimension-minimality. The dashed accepting states are those of $\cA'$ and $\cB'$, and are not accepting in the resulting construction.}

\label{fig:primalityReduction}
\end{figure}
\begin{theorem}
\label{thm:primality and minimality are undecidable}
Primality and dimension-minimality are undecidable, already for $2$-CN.
\end{theorem}
\begin{proof}
We prove the theorem by establishing that $\cC$ is not prime if and only if $\lang(\cA)\subseteq \lang(\cB)$, and $\cC$ is not dimension-minimal if and only if $\lang(\cA)\subseteq \lang(\cB)$.

Assume that $\lang(\cA) \subseteq \lang(\cB)$, then the component of $\cC$ containing $\cA'$ and $\cP$ (\cref{fig:primalityReduction} left) becomes redundant.
Since the component containing $\cB'$ and $q_\top$ only makes use of one counter, $\cC$ is composite.
Formally, we claim that $\lang(\cC) =\set{ u\$v \mid u \in \lang(\cB) \wedge v \in \Lambda^*}$. 
Indeed, if $w\in \lang(\cC)$ then $w=u\$v$ so either $u \in \lang(\cA') = \lang(\cA)$ or $u \in \lang(\cB)$, but since $\lang(\cA)\subseteq \lang(\cB)$, this is equivalent to $u \in \lang(\cB)$, and in this case there is simply no condition on $v \in \Lambda^*$.
Since the second counter is not used in component containing $\cB'$ and $q_\top$ (\cref{fig:primalityReduction} right), we can construct a $1$-CN equivalent to $\cC$ by projecting on the first counter and just deleting the component containing $\cA'$ and $\cP$ completely.
It follows that in this case $\cC$ is not dimension-minimal, and therefore is not prime either.

For the converse, assume that $\lang(\cA) \not\subseteq \lang(\cB)$, and let $u \in \lang(\cA) \setminus \lang(\cB)$. 
Denote $m = \max\set{ \effect{\rho} \mid \rho \text{ is an accepting run of } \cA \text{ on } u }$. 
Thus, for a word $v \in \Lambda^*$ we have that $u\$v \in \lang(\cC)$ if and only if $v$ is accepted in $\cP$ with initial counter $m$. 
Assume by way of contradiction that $\cC$ is not prime, then we can write $\lang(\cC)$ as an intersection of languages of $1$-CNs. 
Loosely speaking, this will create a contradiction as we will be able to argue that $\cP$ is not prime. 
More precisely, take $v = a^{m_1} \# a^{m_2} \# \cdots \# a^{m_{k+1}} \# b^{m_b} c^{m_c}$ for integers $\left\{ m_i \right\}_{i=1}^{k+1}, m_b, m_c \in \bbN$ and consider words of the form $u\$v$. 
Our analysis from~\cref{sec:prime-2vass}---specifically the arguments used in the proof~\cref{lem:two_segments_one_bad}---on $u\$v$ can show, mutatis mutandis, that the language of $\cP$ is not composite regardless of any fixed initial counter value (an analogue of~\cref{thm:2CN_prime}).

We thus have that $\cC$ is prime, and in particular $\cC$ is dimension-minimal, concluding the correctness of the reduction.
\hfill\qed\end{proof}

%In order to slightly contrast the unfortunate undecidability in~\cref{thm:primality and minimality are undecidable}, we show that regularity (i.e., minimisation to dimension $0$, or equivalently -- primality to dimension $0$) is decidable for $k$-DCN. 
To contrast the undecidability of primality in nondeterministic CNs, we turn our attention to a decidable fragment of primality, for which we focus on deterministic CNs.
Recall that by~\cref{prop:DCN_Projection}, a $k$-DCN is dimension minimal if and only if it is not $(1, k-1)$-composite.
Thus, dimension-minimality ``captures'' primality.
We show that regularity, which is equivalent to being $(0, 1)$-composite, is decidable for $k$-DCNs for every dimension $k$.

For dimension one, regularity is already known to be decidable in \class{EXPSPACE}, even for history-deterministic 1-CNs~{\cite[Theorem 19]{bose2023history}}.
History-determinism is a restricted form of nondeterminism; history-deterministic CNs are less expressive than nondeterministic CNs but more expressive than DCNs.
However, already for $k\geq 2$, regularity is undecidable for history-deterministic $k$-CNs~{\cite[Theorem 20]{bose2023history}}.

\begin{theorem}
    \label{thm:DCN_regularity_decidable}
    Regularity of $k$-DCN is decidable and is in \class{EXPSPACE}.
\end{theorem}

We provide further details, including a proof of~\cref{thm:DCN_regularity_decidable},  in~\cref{apx:DCN regularity}.
In short, we translate our $k$-DCN into a regularity preserving Vector Addition System (VAS) and use results on VAS regularity from~\cite[Theorem 4.5]{BlockeletS11}. 
We remark that an alternative approach may be taken by adapting the results of~\cite{Demri13} on regularity of VASS, although this seems more technically challenging because CNs have accepting states. 

% %\section{Decidable Fragments in the Deterministic Setting}
% %\label{sec:decidable_DCN}
% We now turn our attention to decidable fragments of primality. A natural candidate for decidability is the deterministic fragment. 
% Recall that by~\cref{prop:DCN_Projection}, a $k$-DCN is dimension-minimal if and only if it is $(1,k)$-composite. 
% Thus, dimension-minimality ``captures'' primality, and is our focus therefore. 
% Following, we show that regularity, equivalent to being $(0,1)$-composite, is decidable for $k$-DCN, and being $(1,1)$-composite is decidable for $2$-DCN.
% \shtodo{make sure we actually show both}

\section{Expressiveness Trade-Offs between Dimensions and Nondeterminism}
\label{sec:dimension_vs_nondet}
\cref{thm:2CN_prime} implies that 2-CNs are more expressive than 1-CNs, and that nondeterministic models are more expressive than deterministic ones. 
In particular, a $k$-DCN can be decomposed by projection (\cref{prop:DCN_Projection}), and have decidable regularity (\cref{thm:DCN_regularity_decidable}).
It is therefore interesting to study the interplay between increasing the dimension and introducing nondeterminism. 
In this section we present two results: first, we show that dimension and nondeterminism, are incomparable notions, in a sense. Second, we show that increasing the dimension strictly increases expressiveness, for both CNs and DCNs. We remark that the latter may seem like an intuitive and simple claim. 
However, to the best of our knowledge it has never been proved, and moreover, it requires a nontrivial approach to pumping with several counters.

We start by showing that nondeterminism can sometimes compensate for low dimension.
Let $k \in \bbN$ and $\Sigma = \set{ a_1, \ldots, a_k, b_1, \ldots, b_k, c }$; consider the language $L_k =\set{ a_1^{n_1} a_2^{n_2} \cdots a_k^{n_k} b_i c^m \mid i \in [k] \wedge n_i \geq m }$.
It is easy to construct a $k$-DCN as well as a $1$-CN for $L_k$, as depicted by~\cref{fig:k_DCN_with_1CN_but_no_k-1_DCN,fig:1_CN_for_k_DCN_with_1CN_but_no_k-1_DCN} for $k=3$. 
To construct a 1-CN we guess which $b_i$ will be later read, and verify the guess using the single counter in the $a_i^{n_i}$ part. 

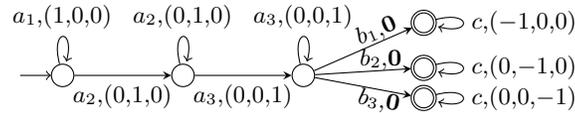
\begin{figure}[ht]
    \centering
    \begin{tikzpicture}[auto,node distance=1.5cm,scale=1]
        \footnotesize
        \thinmuskip=0mu
            \node (q1) [initial, state, initial text = {},inner sep=3pt, minimum size=5pt] {};
            \node (q2) [state,inner sep=3pt, minimum size=3pt] at (1.6,0) {};
            \node (q3) [state,inner sep=3pt, minimum size=5pt] at (3.2,0) {};
            \node (s1) [accepting, state,inner sep=3pt, minimum size=5pt] at (4.8,0.7) {};
            \node (s2) [accepting, state,inner sep=3pt, minimum size=5pt] at (4.8,0.1) {};
            \node (s3) [accepting, state,inner sep=3pt, minimum size=5pt] at (4.8,-0.3) {};
            \path [-stealth]
            (q1) edge [bend left=0] node[below] {$a_2,(0,1,0)$} (q2)
            (q2) edge [bend left=0] node[below] {$a_3,(0,0,1)$} (q3)
            (q3) edge [bend left=0] node[above=-3pt,sloped,pos=0.7] {$b_1,\vec{0}$} (s1)
            (q3) edge [bend left=0] node[above=-3pt,pos=0.7] {$b_2,\vec{0}$} (s2)
            (q3) edge [bend left=0] node[below=-3pt,sloped,pos=0.7] {$b_3,\vec{0}$} (s3)
            (q1) edge [loop above] node {$a_1,(1,0,0)$} (q1)
            (q2) edge [loop above] node {$a_2,(0,1,0)$} (q2)
            (q3) edge [loop above] node {$a_3,(0,0,1)$} (q3)
            (s1) edge [loop right] node {$c,(-1,0,0)$} (s1)
            (s2) edge [loop right] node {$c,(0,-1,0)$} (s2)
            (s3) edge [loop right] node {$c,(0,0,-1)$} (s3);
        \end{tikzpicture}
    \caption{A $3$-DCN for $L_3=\{a_1^{n_1}a_2^{n_2}a_3^{n_3}b_ic^m\mid i\in [3]\wedge n_i\ge m\}$. Intuitively, the $3$-DCN counts the number of occurrences of each letter, and decreases the appropriate counter once the letter $b_i$ selects it.}
    \label{fig:k_DCN_with_1CN_but_no_k-1_DCN}  
\end{figure}

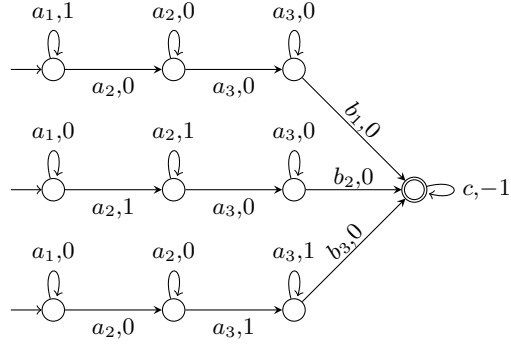
\begin{figure}[ht]
    \centering
    \begin{tikzpicture}[auto,node distance=1.5cm,scale=1]
        \footnotesize
        \thinmuskip=0mu
            \node (q11) [initial, state, initial text = {},inner sep=3pt, minimum size=5pt] {};
            \node (q12) [state,inner sep=3pt, minimum size=3pt] at (1.6,0) {};
            \node (q13) [state,inner sep=3pt, minimum size=5pt] at (3.2,0) {};
            \node (q21) [initial, state, initial text = {},inner sep=3pt, minimum size=5pt] at (0,-1.6) {};
            \node (q22) [state,inner sep=3pt, minimum size=3pt] at (1.6,-1.6) {};
            \node (q23) [state,inner sep=3pt, minimum size=5pt] at (3.2,-1.6) {};
            \node (q31) [initial, state, initial text = {},inner sep=3pt, minimum size=5pt] at (0,-3.2) {};
            \node (q32) [state,inner sep=3pt, minimum size=3pt] at (1.6,-3.2) {};
            \node (q33) [state,inner sep=3pt, minimum size=5pt] at (3.2,-3.2) {};
            \node (s) [accepting, state,inner sep=3pt, minimum size=5pt] at (4.8,-1.6) {};
            \path [-stealth]
            (q11) edge [bend left=0] node[below] {$a_2,0$} (q12)
            (q12) edge [bend left=0] node[below] {$a_3,0$} (q13)
            (q21) edge [bend left=0] node[below] {$a_2,1$} (q22)
            (q22) edge [bend left=0] node[below] {$a_3,0$} (q23)
            (q31) edge [bend left=0] node[below] {$a_2,0$} (q32)
            (q32) edge [bend left=0] node[below] {$a_3,1$} (q33)
            (q11) edge [loop above] node {$a_1,1$} (q11)
            (q12) edge [loop above] node {$a_2,0$} (q12)
            (q13) edge [loop above] node {$a_3,0$} (q13)
            (q21) edge [loop above] node {$a_1,0$} (q21)
            (q22) edge [loop above] node {$a_2,1$} (q22)
            (q23) edge [loop above] node {$a_3,0$} (q23)
            (q31) edge [loop above] node {$a_1,0$} (q31)
            (q32) edge [loop above] node {$a_2,0$} (q32)
            (q33) edge [loop above] node {$a_3,1$} (q33)
            (q13) edge [bend left=0] node[above=-3pt,sloped] {$b_1,0$} (s)
            (q23) edge [bend left=0] node[above=-3pt,sloped] {$b_2,0$} (s)
            (q33) edge [bend left=0] node[above=-3pt,sloped] {$b_3,0$} (s)
            (s) edge [loop right] node {$c,-1$} (s);
        \end{tikzpicture}
    \caption{A $1$-CN for $L_3=\{a_1^{n_1}a_2^{n_2}a_3^{n_3}b_ic^m\mid i\in [3]\wedge n_i\ge m\}$. Intuitively, the $1$-CN guesses which $b_i$ will be seen, and counts the respective occurrences of the letter $a_i$. Then, once $b_i$ is seen, the counter is decreased on $c$.}
    \label{fig:1_CN_for_k_DCN_with_1CN_but_no_k-1_DCN}  
\end{figure}
However, we show in~\cref{apx:k_DCN_no_k-1_DCN} that the dimension cannot be minimised whilst maintaining determinism.
\begin{proposition}
\label{prop:k_DCN_no_k-1_DCN}
$L_k$ is not recognisable by a $(k-1)$-DCN.
\end{proposition}
% \begin{proof}
% Assume by way of contradiction that there exists a $(k-1)$-DCN $\cD = \tup{ \Sigma, Q, Q_0, \delta, F}$ such that $\lang(\cD) = L_k$.
% Let $n > |Q|$ and for every $i \in [k]$ consider the word $w_i = a_1^n a_2^n \cdots a_k^n b_i c^n \in L_k$. 
% Since $\cD$ is deterministic and $n > |Q|$, all of the accepting runs on the $w_i$ coincide up to the $b_i$ part and have cycles in each $a_i^n$ segment as well as in the $c^n$ segment (the latter may differ according to $i$). 
% Let $M$ be the product of the lengths of all these cycles.

% First, observe that the cycles in all of the $a_i^n$ segments cannot decrease any counter. 
% Indeed, otherwise by pumping such a cycle for large enough $t>0$ times, there would not exist an $\bbN$-run on words with the prefix $a_1^n \cdots a_{i-1}^n a_i^{n+tM}$.
% This creates a contradiction since, with an appropriate suffix, such words can be accepted.

% Thus, all $a_i$ cycles have non-negative effects for all counters. 
% Indeed, for each counter $i$ -- associate with $i$ the minimal segment index whose cycle strictly increases $i$. 
% Since there are $k-1$ counters and $k$ segments this map is not surjective, in other words, there is a segment (without loss of generality, the $a_k$ segment) such that every counter that is increased in the $a_k$ cycle is also increased in a previous segment. 
% Therefore, there exist $s,t>0$ such that the word 
% \begin{equation*}
%     a_1^{n+sM}a_2^{n+sM}\cdots a_{k-1}^{s+sM}a_k^nb_kc^{n+tM} \notin L_k
% \end{equation*} 
% is accepted by $\cD$, which is a contradiction. 
% \hfill\qed\end{proof}

We now turn to show that conversely, dimension can sometimes compensate for nondeterminism. 
Moreover, we show that there is a strict hierarchy of expressiveness with respect to dimension. Specifically, for $k\in \bbN$ consider the language 
$H_k=\{a_1^{m_1}a_2^{m_2}\cdots a_k^{m_k}b_1^{n_1}b_2^{n_2}\cdots b_k^{n_k}\mid \forall 1\le i\le k,\ m_i\ge n_i\}$.
\begin{theorem}
\label{thm:dimension hierarchy language}
    $H_k$ is recognisable by a $k$-DCN, but not by a $(k-1)$-CN.
\end{theorem}
\begin{proof}[sketch]
     Constructing a $k$-DCN for $H_k$ is straightforward, by using the $i$-th counter to check that $m_i\ge n_i$, for each $i\in [k]$.

    We turn to argue that $H_k$ is not recognisable by a $(k-1)$-CN. See~\cref{apx:dimension hierarchy language} for the full proof. 
    Assume by way of contradiction that $\cA=\tup{\Sigma,Q,Q_0,\delta,F}$ is a $(k-1)$-CN with $L(\cA)=H_k$. We first observe that there exists $m_1\in \bbN$ large enough so that every run of $\cA$ on $a_1^{m_1}$ must traverse a non-negative cycle, i.e., a cycle whose overall effect is $\vec{u_1}\in \bbZ^{k-1}$ such that $\vec{u_1}[i]\ge 0$ for all $i\in [k-1]$. 
    Indeed, this is immediate by a ``uniformly bounded'' version of Dickson's lemma~\cite{figueira2011ackermannian}; any long-enough ``controlled'' sequence of vectors in $\bbN^{k-1}$ must contain an $r$-increasing chain, for any $r\in \bbN$. 
    %By choosing $r=|Q|+1$, we obtain a cycle whose overall effect is non-negative.

    By repeating this argument we can ultimately find $m_1,\ldots,m_k$ such that any run of $\cA$ on $a_1^{m_1}a_2^{m_2}\cdots a_k^{m_k}$ traverses a non-negative cycle in each $a_j$ segment for $j\in [k]$.
    Consider now the word $w=a_1^{m_1}a_2^{m_2}\cdots a_k^{m_k}b_1^{m_1}b_2^{m_2}\cdots b_k^{m_k}\in H_k$, then there exists an accepting run $\rho$ of $\cA$ on $w$ such that for each $j\in [k]$, the run $\rho$ traverses a non-negative cycle in segment $a_j$, with effect $\vec{u_j}\in \bbN^{k-1}$. 

    Consider the vectors $\vec{u_1},\ldots,\vec{u_k}$. We claim that there exists $\ell\in [k]$ such that the support of $\vec{u_\ell}$ is \emph{covered} by $\vec{u_1},\ldots,\vec{u_{\ell-1}}$ in the following sense: for every counter $i\in [k-1]$, if $\vec{u_\ell}[i]>0$, then there exists $j<\ell$ such that $\vec{u_j}[i]>0$. 
    Indeed, this holds since otherwise every $\vec{u_j}$ must contribute a fresh positive coordinate to the union of supports of the previous vectors, but there are $k$ vectors and only $k-1$ coordinates. 

    Next, observe that since each $\vec{u_j}$ is a non-negative cycle taken in $\rho$, then it can be pumped without decreasing any following counters, and hence induce an accepting run on a pumped word.
    Intuitively, we now proceed by pumping all the $\vec{u_j}$ cycles for $j<\ell$ for some large-enough number of times $M$, which enables us to remove one iteration of the cycle with effect $\vec{u_\ell}$ while maintaining an accepting run on a word of the form: 
    %Formally, for every $j\in [k]$ denote by $d_j$ the length of a cycle in the $a_j^{m_j}$ segment whose effect is $\vec{u^j}$. For every counter $i$ for which $(\vec{u^{\ell}})_i>0$, choose $f(i)<\ell$ with $(\vec{u^{f(i)}})_{i}>0$ (which exists since $\vec{u^{\ell}}$ is covered as shown above). Let $M\in \bbN$ be large enough so that $M \cdot (\vec{u^{f(i)}})_{i}>(\vec{u^{\ell}})_i$  for all $i\in [k-1]$. 
    %We then obtain an accepting run of $\cA$ on the word
    \[
    w'=a_1^{m_1+Md_1}a_2^{m_2+Md_2}\cdots a_{\ell-1}^{m_{\ell-1}+Md_{\ell-1}} a_\ell^{m_\ell-d_\ell}a_{\ell+1}^{m_{\ell+1}}\cdots a_k^{m_k}b_1^{m_1}b_2^{m_2}\cdots b_k^{m_k}.
    \]
    Since $m_\ell>m_\ell-d_\ell$, the $b_\ell$ segment is longer than the $a_\ell$ segment.
    Thus $w'\notin H_k$, this yields a contradiction.
 %   
    %by repeating the cycle inducing $\vec{u^j}$ additional $M$ times (hence yielding the added $Md_j$ length) for $j<\ell$, and removing one iteration of the cycle inducing $\vec{u^\ell}$. Since the counter values leading up to the removal of $\vec{u^\ell}$ were increased by at least their value after adding $\vec{u^\ell}$, the remainder of the run can be taken without change.
%
    %Finally, note that $w'\notin H_k$, since $m_\ell-d_\ell<m_\ell$, with the former being the length of the $a_\ell$ segment and the latter the length of the $b_\ell$ segment.
    \hfill\qed
\end{proof}
Apart from showing that nondeterminism cannot always compensate for increased dimension, \cref{thm:dimension hierarchy language} also shows that for every dimension $k$, there are languages recognisable by a $(k+1)$-DCN (and in particular by a $(k+1)$-CN), but not by any $k$-CN (and in particular not by any $k$-DCN). 
Thus, we obtain the following hierarchy.
\begin{corollary}
    \label{cor:expressiveness_hierarchy_CNs}
    For every $k\in \bbN$, $k$-CNs (resp. $k$-DCNs) are strictly less expressive than $(k+1)$-CNs (resp. $(k+1)$-DCNs).
\end{corollary}

\section{Discussion}
\label{sec:discussion}
Broadly, this work explores the interplay between the dimension of a CN and its expressive power. This is done by studying the \emph{dimension-minimality} problem, where we ask whether the dimension of a given CN can be decreased while preserving its language, and by the more involved \emph{primality} problem, which allows a decomposition to multiple CNs of lower dimension. 
We show that both primality and dimension-minimality are undecidable. Moreover, they remain undecidable even when we discard the degenerate dimension 0 case, which corresponds to finite memory, i.e., regular languages. On the other hand, this degenerate case is one where we can show decidability for DCNs.

This work also highlights a technical shortcoming of current understanding of high-dimensional CNs: pumping arguments in the presence of $k$ dimensions and nondeterminism are very involved, and are (to our best efforts) insufficient to prove \cref{conj:prime kCN}.
To this end, we present novel pumping arguments in the proof of~\cref{thm:2CN_prime} and to some extent in the proof of~\cref{thm:dimension hierarchy language}, which make progress towards pumping in the presence of $k$ dimensions and nondeterminism.

% We conclude with two open problems for future research.

% \noindent{\bf Example of a prime $k$-CN:} 
% We demonstrate a prime $2$-CN, and conjecture a family of prime $k$-CN in \cref{conj:prime kCN}. More generally, we can ask whether for every $d<k\in \bbN$ there is a $k$-CN that is $(d, n)$-composite but not $(d-1, m)$-composite, for some $n, m \in \bbN$?
% The difficulty involved in proving~\cref{thm:2CN_prime} seems to indicate that this might be highly nontrivial.

% \noindent{\bf Decidability of dimension-minimality for $k$-DCN:} 
% \shtodo{remove this if we drop regularity}
% In~\cref{thm:DCN_regularity_decidable} we show that regularity is decidable for DCN, but this relies heavily on the finite-memory property of regular languages. Extending this to deciding whether a $k$-DCN has an equivalent $d$-DCN for $d<k$ seems to be technically challenging, even for $k=2$ and $d=1$.

% \shtodo{add something about dimension minimization, if we leave it as a conjecture.}

\bibliography{main_FoSSaCS}

\appendix
    \section{Missing Proofs}
    \subsection{Proof of~\cref{prop:DCN_dimension_minimal_iff_composite}}
    \label{apx:DCN_dimension_minimal_iff_composite}
        If $\cD$ is $(1,k-1)$-composite, then there exist $1$-DCNs $\cB_1,\ldots,\cB_{k-1}$ such that $\lang(\cD)=\bigcap_{i=1}^{k-1} \lang(\cB_i)$. Define $\cB$ to be the product $(k-1)$-CN of $\cB_1,\ldots,\cB_{k-1}$, then $\lang(\cB)=\lang(\cD)$, so $\cD$ is not dimension-minimal.

    Conversely, if $\cD$ is not dimension-minimal, there exists w.l.o.g., a $k-1$-CN $\cB$ such that $\lang(\cB)=\lang(\cD)$. Then, we have $\lang(\cB)=\bigcap_{i=1}^{k-1} \lang(\cB|_i)$, so $\cD$ is $(1,k-1)$-composite.\qed 

     \subsection{Proof of~\cref{lem:maxcounternopositivecycles}}
        \label{apx:lem:maxcounternopositivecycles}
        We prove the contra-positive: assume $\rho$ has counter value $n+W|Q|$, w.l.o.g. we can assume this occurs at the end. Since the maximal counter increase along a transition is $W$, it follows that there are at least $|Q|$ indices where the counter increases beyond any previous level. Thus, we can find indices $i_0,i_1,\ldots,i_{|Q|}$ such that the counter increases in the infix of $\rho$ between $i_j$ and $i_{j+1}$. Since these are $|Q|+1$ states, then there exists a state that is visited twice, which forms a $\positive$ cycle in $\rho$.
        \qed

\subsection{Proof of~\cref{prop:k_DCN_no_k-1_DCN}}
\label{apx:k_DCN_no_k-1_DCN}
Assume by way of contradiction that there exists a $(k-1)$-DCN $\cD = \tup{ \Sigma, Q, Q_0, \delta, F}$ such that $\lang(\cD) = L_k$.
Let $n > |Q|$ and for every $i \in [k]$ consider the word $w_i = a_1^n a_2^n \cdots a_k^n b_i c^n \in L_k$. 
Since $\cD$ is deterministic and $n > |Q|$, all of the accepting runs on the $w_i$ coincide up to the $b_i$ part and have cycles in each $a_i^n$ segment as well as in the $c^n$ segment (the latter may differ according to $i$). 
Let $M$ be the product of the lengths of all these cycles.

First, observe that the cycles in all of the $a_i^n$ segments cannot decrease any counter. 
Indeed, otherwise by pumping such a cycle for large enough $t>0$ times, there would not exist an $\bbN$-run on words with the prefix $a_1^n \cdots a_{i-1}^n a_i^{n+tM}$.
This creates a contradiction since, with an appropriate suffix, such words can be accepted.

Thus, all $a_i$ cycles have non-negative effects for all counters. 
Indeed, for each counter $i$ -- associate with $i$ the minimal segment index whose cycle strictly increases $i$. 
Since there are $k-1$ counters and $k$ segments this map is not surjective, in other words, there is a segment (without loss of generality, the $a_k$ segment) such that every counter that is increased in the $a_k$ cycle is also increased in a previous segment. 
Therefore, there exist $s,t>0$ such that the word 
\begin{equation*}
    a_1^{n+sM}a_2^{n+sM}\cdots a_{k-1}^{s+sM}a_k^nb_kc^{n+tM} \notin L_k
\end{equation*} 
is accepted by $\cD$, which is a contradiction. 
\qed

\subsection{Proof of~\cref{thm:dimension hierarchy language}}
\label{apx:dimension hierarchy language}
 Constructing a $k$-DCN for $H_k$ is straightforward, by using the $i$-th counter to check that $m_i\ge n_i$, for each $i\in [k]$.%, see~\cref{fig:DCN-for-H}.

    We turn to prove that $H_k$ is not recognisable by a $(k-1)$-CN. Assume by way of contradiction that $\cA=\tup{\Sigma,Q,Q_0,\delta,F}$ is a $(k-1)$-CN with $L(\cA)=H_k$. We first observe that there exists $m_1\in \bbN$ large enough so that every run of $\cA$ on $a_1^{m_1}$ must traverse a non-negative cycle, i.e., a cycle whose overall effect is $\vec{u_1}\in \bbZ^{k-1}$ such that $\vec{u_1}[i]\ge 0$ for all $i\in [k-1]$. 
    Indeed, this is immediate by a ``uniformly bounded'' version of Dickson's lemma~\cite{figueira2011ackermannian}; any long-enough ``controlled'' sequence of vectors in $\bbN^{k-1}$ must contain an $r$-increasing chain, for any $r\in \bbN$.  
    By choosing $r=|Q|+1$, we obtain a cycle whose overall effect is non-negative.

    By repeating this argument, taking the maximal possible counter value attained after $a_1^{m_1}$ as the initial point of the controlled sequence, we can also compute $m_2$ such that any run of $\cA$ on $a_1^{m_1}a_2^{m_2}$ must traverse a non-negative cycle also in the $a_2$ segment of the run. Similarly, we can ultimately find $m_1,\ldots,m_k$ such that any run of $\cA$ on $a_1^{m_1}a_2^{m_2}\cdots a_k^{m_k}$ traverses a non-negative cycle in each $a_j$ segment for $j\in [k]$.

    Consider now the word $w=a_1^{m_1}a_2^{m_2}\cdots a_k^{m_k}b_1^{m_1}b_2^{m_2}\cdots b_k^{m_k}\in H_k$, then there exists an accepting run $\rho$ of $\cA$ on $w$ such that for each $j\in [k]$, the run $\rho$ traverses a non-negative cycle in segment $a_j$, with effect $\vec{u_j}\in \bbN^{k-1}$. 

    Consider the vectors $\vec{u_1},\ldots,\vec{u_k}$. We claim that there exists $\ell\in [k]$ such that the support of $\vec{u_\ell}$ is \emph{covered} by $\vec{u_1},\ldots,\vec{u_{\ell-1}}$ in the following sense: for every counter $i\in [k-1]$, if $\vec{u_\ell}[i]>0$, then there exists $j<\ell$ such that $\vec{u_j}[i]>0$. 
    Indeed, this holds since otherwise every $\vec{u_{j}}$ must contribute a fresh positive coordinate to the union of supports of the previous vectors, but there are $k$ vectors and only $k-1$ coordinates. 

    Next, observe that since each $\vec{u_j}$ is a non-negative cycle taken in $\rho$, then it can be pumped without decreasing any following counters, and hence induce an accepting run on a pumped word.
    Intuitively, we now proceed by pumping all the $\vec{u_j}$ cycles for $j<\ell$ enough times to enable us to remove one iteration of the cycle $\vec{u_\ell}$ while maintaining an accepting run. Since the number of `$b_\ell$'s remains $m_\ell$, this is yields a contradiction.
    
    Formally, for every $j\in [k]$ denote by $d_j$ the length of a cycle in the $a_j^{m_j}$ segment whose effect is $\vec{u_j}$. For every counter $i$ for which $\vec{u_{\ell}}[i]>0$, choose $f(i)<\ell$ with $\vec{u_{f(i)}}[i]>0$, such an $f(i)$ exists given that $\vec{u_{\ell}}$ is covered, as detailed above. 
    Now, let $M\in \bbN$ be large enough such that $M \cdot \vec{u_{f(i)}}[i]>\vec{u_\ell}[i]$  for all $i\in [k-1]$. We then obtain an accepting run of $\cA$ on the word
    \[
    w'=a_1^{m_1+Md_1}a_2^{m_2+Md_2}\cdots a_{\ell-1}^{m_{\ell-1}+Md_{\ell-1}} a_\ell^{m_\ell-d_\ell}a_{\ell+1}^{m_{\ell+1}}\cdots a_k^{m_k}b_1^{m_1}b_2^{m_2}\cdots b_k^{m_k}
    \]
    by repeating the cycle inducing $\vec{u_j}$ additional $M$ times (hence the additional $Md_j$ length) for $j<\ell$, and removing one iteration of the cycle inducing $\vec{u_\ell}$. 
    Since the counter values leading up to the removal of $\vec{u_\ell}$ were increased by at least their value after adding $\vec{u_\ell}$, the remainder of the run can be taken without change.

    Finally, since $m_\ell-d_\ell<m_\ell$, the former being the length of the $a_\ell$ segment and the latter the length of the $b_\ell$ segment, we know that $w'\notin H_k$.
\hfill\qed

\section{Greyscale-Friendly Version of~\cref{sec:prime-2vass}}
    \label{apx:grey}
    Here, we repeat~\cref{sec:bad_segments,sec:proof_of_2Prime} with augmented symbols for positive and non-negative cycles.
    \renewcommand{\posCol}[1]{{\color{cbOne}{\widehat{#1}}}}
    \renewcommand{\nonnegCol}[1]{{\color{cbTwo}{\widetilde{#1}}}}
    
    \subsection{Good and Bad Segments}
    \label{apx:sec:bad_segments}
    We lift the colour scheme of $\positive$ and $\nonnegative$ to words and runs as follows. For a word $w=uv$ and a run $\rho$, we write e.g., $\posCol{u}\nonnegCol{v}$ to denote that $\rho$ traverses a $\positive$ cycle when reading $u$, then a $\nonnegative$ cycle when reading $v$. Note that this does not preclude other cycles, e.g., there could also be negative cycles in the $u$ part, etc. That is, the colouring is not unique, but represents elements of the run.

Recall our assumption that $\lang(\cP)=\bigcap_{1 \leq j \leq k}\lang(\cV_i)$, and denote $\cV_j=\tup{\Sigma,Q^j,I^j,\delta^j,F^j}$ for all $j\in [k]$. Let $Q_{\max}=\max\{|Q_j|\}_{j=1}^k$ and denote  $\alpha=Q_{\max}!$. 
Further recall that we focus on words of the form $a^{m_1}\#a^{m_2}\#\cdots\#a^{m_{k+1}}\#b^{m_b}c^{m_c}$ for integers $\left\{m_i\right\}_{i=1}^{k+1},m_b,m_c \in \bbN$, and that we refer to the infix $a^{m_i}$ as Segment $i$, for $1 \leq i \leq k$.
We proceed to formally define good and bad segments.
\begin{definition}[Good and Bad Segments]
\label{apx:def:good_bad}
Segment $i$ is \emph{bad} in $\cV_j$ if there exist constants $\left\{m_i\right\}_{i=1}^{k+1},m_b,m_c \in \bbN$ such that all the following hold.
\begin{itemize}
    \item $\left\{m_i\right\}_{i=1}^{k+1},m_b,m_c$ are multiples of $\alpha$. 
    \item There is an accepting $\bbN$-run $\rho$ of $\cV_j$ on $w=a^{m_1}\#a^{m_2}\#\cdots\#a^{m_{k+1}}\#b^{m_b}c^{m_c}$ that adheres to one of the following three \emph{forms}:
    \begin{enumerate}
        \item $\nonnegCol{a^{m_1}}\#\nonnegCol{a^{m_2}}\#\cdots \nonnegCol{a^{m_{i-1}}}\#\posCol{a^{m_i}}\#a^{m_{i+1}}\#\cdots\#a^{m_{k+1}}\#b^{m_b}c^{m_c}$
        \item $\nonnegCol{a^{m_1}}\#\nonnegCol{a^{m_2}}\#\cdots \nonnegCol{a^{m_{i-1}}}\#\nonnegCol{a^{m_i}}\#\nonnegCol{a^{m_{i+1}}}\#\cdots\#\nonnegCol{a^{m_{k+1}}}\#\nonnegCol{b^{m_b}}\nonnegCol{c^{m_c}}$
        \item $\nonnegCol{a^{m_1}}\#\nonnegCol{a^{m_2}}\#\cdots \nonnegCol{a^{m_{i-1}}}\#\nonnegCol{a^{m_i}}\#\nonnegCol{a^{m_{i+1}}}\#\cdots\#\nonnegCol{a^{m_{k+1}}}\#\posCol{b^{m_b}}c^{m_c}$.
    \end{enumerate}
\end{itemize}

Segment $i$ is \emph{good} in $\cV_j$ if it is not bad.
\end{definition}

\cref{apx:lem:bad_segements_are_bad} below formalises the intuition that a bad segment can be pumped simultaneously with both the $b$ and $c$ segments, eventually generating a word that is not in $\lang(\cP)$, but is accepted by $\cV_j$.

Intuitively, Forms 2 and 3 mean that all segments are bad. Indeed, Segment $i$ has a $\nonnegative$ cycle, so it can be pumped safely, and in Form 2 both $b$ and $c$ can be pumped using $\nonnegative$ cycles, whereas in Form 3 we can pump $b$ using a $\positive$ cycle, and use it to compensate for pumping $c$, even if the latter requires pumping a negative cycle.

Form 1 is the interesting case, where we use a $\positive$ cycle in Segment $i$ to compensate for pumping both $b$ and $c$. The requirement that all segments up to $i$ are $\nonnegative$ is at the heart of our proof, and is explained in~\cref{apx:sec:proof_of_2Prime}.
Formally, we have the following.

\begin{lemma}
    \label{apx:lem:bad_segements_are_bad}
    Let $l$ be a bad segment in $\cV_j$, then there exist $x,y,z\in \bbN$ multiples of $\alpha$ such that for every $n\in \bbN$ the following word $w$ is accepted by $\cV_j$.
    \[w_n=a^{m_1}\#a^{m_2}\#\cdots\#a^{m_l+xn}\#a^{m_{l+1}}\#\cdots\#a^{m_{k+1}}\#b^{m_b+yn}c^{m_c+zn}\]
\end{lemma}
\begin{proof}
    We can choose $z=\alpha$, then take $y$ to be large enough so that Form 3 runs can compensate for negative cycles in $c^z$ using $\positive$ cycles in $b^y$, while not decreasing the counters in Form $2$ runs (we can find such $y$ that is a multiple of $\alpha$, since $\alpha$ is divisible by all lengths of simple cycles). 
    Finally, we choose $x$ so that Form 1 runs can compensate for $c^z$ and $b^y$ using $\positive$ cycles in $a^x$ in Segment $l$, again while not decreasing the counters in Forms $2$ and $3$.
\hfill\qed\end{proof}

Recall that our goal is to show that there is a segment $l \in [k+1]$ that is bad in \emph{every} $V_j$, for $j \in [k]$.~\cref{apx:lem:two_segments_one_bad} below shows that each $V_j$ has at most one good segment. Therefore, there are at most $k$ good segments in total, leaving at least one bad segment as desired.

\begin{lemma}
\label{apx:lem:two_segments_one_bad}
    Let $j\in [k]$ and $r< s\in [k+1]$. Then either Segment $r$ or $s$ is bad in $\cV_j$.
\end{lemma}
\begin{proof}
    Since $j$ is fixed, denote $\cV_j=\tup{\Sigma,Q,Q_0,\delta,F}$. We inductively define constants $\left\{n_i\right\}_{i=1}^{k+1},n_b,n_c \in \bbN$ as follows. $n_1$ is a large-enough multiple of $\alpha$ so that~\cref{lem:largeenoughwordnonnegativecycle} guarantees a $\nonnegative$ cycle in any accepting run of $\cV_j$ on $a^{n_1}$ from some $(q_0,0)$ with $q_0\in Q_0$. 
    Assume we have defined $n_1,\ldots n_{l-1}$, and consider the word $u=a^{n_1}\#a^{n_2}\#\cdots\#a^{n_{l-1}}\#$. Define $n=|u|W$ where $W$ is the maximal value in any transition of $\cV_j$.
    Since $u$ consists of $\frac{n}{W}$ letters, $n+1$ is greater than any counter value that can be reached by $\cV_j$ by reading $u$. We define $n_l$ to be a multiple of $\alpha$ large enough so that~\cref{lem:largeenoughwordnonnegativecycle} guarantees a $\nonnegative$ cycle when reading $a^{n_l}$ from any configuration of the form $\{(q,n') \mid q\in Q,\ n\le n+1\}$. We set $n_b=n_c=\alpha$ (the choice of $n_b,n_c$ is slightly arbitrary). Finally, we set $w=a^{n_1}\#\cdots\#a^{n_{k+1}}\#b^{n_b}c^{n_c}$.

    Now, for every $x\in \bbN$, we obtain from $w$ a word $w_x$ by pumping $x\alpha$ $a$'s to segments $r,s$ and to the $b$ and $c$ segments. That is, let $n'_i=n_i+x\alpha$ for $i\in \{r,s\}$ and $n'_i=n_i$ for $i\notin \{r,s\}$, and let $n'_b=n_b+x \alpha$ and $n'_c=n_c+x \alpha$, then $w_x=a^{n'_1}\#\cdots\#a^{n'_{k+1}}\#b^{n'_b}c^{n'_c}$.
    Observe that $w_x\in \lang(\cP)$. Indeed, since $n_{r} \geq n_b=\alpha$ and $n_{s} \geq n_c=\alpha$ we have that $n_{r}+x\alpha \geq n_b+x\alpha$ and $n_{s}+x\alpha \geq n_c+x\alpha$, so segments $r$ and $s$ can already pay for the $b$'s and $c$'s, respectively.
    In particular, $w_x\in \lang(\cV_j)$ with some accepting $\bbN$-run $\rho_x$. 

    We choose a particular value of $x$, as follows. Consider $x$ and suppose some accepting $\bbN$-run $\rho_x$ as above does not traverse a $\positive$ cycle neither in segment $r$ nor in $s$. By~\cref{lem:maxcounternopositivecycles}, the maximal possible counter value of $\rho_x$ after reading \[a^{n_1}\# \cdots\# a^{n_{r}+ x\alpha}\#\cdots\#a^{n_{s}+x\alpha}\#\cdots\#a^{n_{k+1}}\#\] is $M_b=(k+1+\sum_{z\in [k+1]\setminus\{r,s\}}n_z)W+2|Q|W$. Crucially, this value does not depend on $x$. Further, if there is no $\positive$ cycle in the segment of $b$'s as well, again the maximal counter value of $\rho$ up to the $c$ segment is bounded by $M_c=(k+2+\sum_{z\in [k+1]\setminus\{r,s\}}n_z)W+3|Q|W$, independent of $x$ and $M_b$. 
    
    By~\cref{lem:largeenoughwordnonnegativecycle}, we can now choose $x$ large enough such that for every accepting $\bbN$-run $\rho_x$ on $w_x$:
    \begin{enumerate}
        \item If $\rho_x$ does not traverse any $\positive$ cycle in segments $r,s$, then $\rho_x$ has a $\nonnegative$ cycle reading $b^{(n_b+x\alpha)}$ from any configuration of the form $\{(q,M') \mid q\in Q,\ M'\le M_b\}$.
        \item If $\rho_x$ does not traverse any $\positive$ cycle in segment $r$ nor $s$, nor in the $b$ segment, $\rho_x$ has a $\nonnegative$ cycle reading $c^{(n_c+x\alpha)}$ from any configuration of the form $\{(q,M') | q\in Q,\ M'\le M_c\}$. 
    \end{enumerate}

    Having fixed $x$, we claim that one of $r,s$ is bad for the constants in $w_x$.     

    By construction,~\cref{lem:largeenoughwordnonnegativecycle} guarantees that $\rho_x$ has a $\nonnegative$ cycles in segments $1,\ldots r-1$. If $\rho_x$ has a $\positive$ cycle in segment $r$, then $\rho_x$ is of Form 1: 
    \[\nonnegCol{a^{n_1}}\#\nonnegCol{a^{n_2}}\#\cdots\nonnegCol{a^{n_{i-1}}}\#\posCol{a^{n_{r}+x\alpha}}\#\cdots \#a^{n_{s}+x\alpha}\#\cdots\#a^{n_{k+1}}\#b^{n_b+x\alpha}c^{n_c+x\alpha}\]
    so $r$ is bad in $\cV_j$.
  
    If $\rho_x$ does not have a $\positive$ cycles in segment $r$, then again by construction~\cref{lem:largeenoughwordnonnegativecycle} guarantees $\nonnegative$ cycles in segments $r,r+1,\ldots, s-1$. Indeed, for $r$ -- we are guaranteed a $\nonnegative$ cycle reading $a^{n_{r}}$, all the more so for $a^{n_{r}+x\alpha}$. 
    As for $r+1,\ldots s-1$ -- if $\rho_x$ does not have a $\positive$ cycle in segment $r$, then the maximal effect of segment $r$ is $W|Q|$. However, $n_{r+1}$ was constructed to guarantee a $\nonnegative$ cycle even in case the effect of segment $r$ is $Wn_{r} \geq W\alpha \geq W|Q|$. 
    
    If there is a $\positive$ cycle in segment $s$ - then $\rho_x$ adheres to Form 1: 
    \[\nonnegCol{a^{n_1}}\#\nonnegCol{a^{n_2}}\#\cdots \nonnegCol{a^{n_{s-1}}}\#\posCol{a^{n_{s}+x\alpha}}\#a^{n_{s+1}}\#\cdots\#a^{n_{k+1}}\#b^{n_b+x\alpha}c^{m_c+x\alpha}\]
    and $s$ is bad in $\cV_j$. 
    
    Otherwise, using the same arguments as for segment $r$, we have that segments $s+1,\ldots,i_{k+1}$ contain $\nonnegative$ cycles. 
    In this case we are left with the $b$ and $c$ segments. The choice of $x$ guarantees a $\nonnegative$ cycle in the segment of $b$'s. If $\rho_x$ traverses a $\positive$ cycle in the $b$ segment, then $w_x$ is of Form 3: 
    \[\nonnegCol{a^{n_1}}\#\nonnegCol{a^{n_2}}\#\cdots \nonnegCol{a^{n_{k+1}}}\#\posCol{b^{n_b+x\alpha}}c^{n_c+x\alpha}\] 
    Finally, if there are no $\positive$ cycles in the $b$ segment, then the choice of $x$ again guarantees a $\nonnegative$ cycle in the $c$ segment, so $w_x$ is of Form 2: \[\nonnegCol{a^{n_1}}\#\nonnegCol{a^{n_2}}\#\cdots\nonnegCol{a^{n_{k+1}}}\#\nonnegCol{b^{n_b+x\alpha}}\nonnegCol{c^{n_c+x\alpha}}\]
    In the two latter cases both $r$ and $s$ are bad in $\cV_j$.
\hfill\qed\end{proof}
\subsection{Proof of~\cref{thm:2CN_prime}}
\label{apx:sec:proof_of_2Prime}
We now have that each $\cV_j$ has at most one good segment (otherwise we would have two good segments $r$ and $s$, contradicting~\cref{apx:lem:two_segments_one_bad}). 
Therefore, all $1$-CNs $\cV_1,\ldots,\cV_k$ have at most $k$ good segments combined. Recall that our words have $k+1$ segments, and therefore there is at least one segment $l$ that is bad in every $\cV_j$. 
Note, however, that this segment may correspond to different constants in each $\cV_j$. That is, there exists constants $\{m_i^j,m_b^j,m_c^j\mid i\in [k+1],j\in [k]\}$ witnessing that segment $l$ is bad for each $\cV_j$. 
We group the $\cV_j$ according to the Form of their accepting runs $\rho_j$, as follows.
\begin{itemize} 
    \item Form 1 are $\nonnegCol{a^{m^j_1}}\#\nonnegCol{a^{m^j_2}}\#\cdots\#\posCol{a^{m^j_l}}\#a^{m^j_{l+1}}\#\cdots\#a^{m^j_{k+1}}\#b^{m^j_b}c^{m^j_c}$. 
    \item Form 2 are $\nonnegCol{a^{m^j_1}}\#\nonnegCol{a^{m^j_2}}\#\cdots\#\nonnegCol{a^{m^j_l}}\#\nonnegCol{a^{m^j_{l+1}}}\#\cdots\#\nonnegCol{a^{m^j_{k+1}}}\#\nonnegCol{b^{m^j_b}}\nonnegCol{c^{m^j_c}}$. 
    \item Form 3 are $\nonnegCol{a^{m^j_1}}\#\nonnegCol{a^{m^j_2}}\#\cdots\#\nonnegCol{a^{m^j_l}}\#\nonnegCol{a^{m^j_{l+1}}}\#\cdots\#\nonnegCol{a^{m^j_{k+1}}}\#\posCol{b^{m^j_b}}c^{m^j_c}$. 
\end{itemize}
We now find constants such that the resulting (single) word has $l$ as a bad segment in all $\cV_j$.
First, for $i\in [k+1]\setminus\{l\}$, define $M_i=\max\{m^j_i\}_{1\leq j \leq k}$ (note that these are still multiples of $\alpha$). Similarly, define $M_c=\max\{m^j_c\}_{1\leq j \leq k}$. It remains to fix constants $\newl$ and $\newb$, which we do in phases in the following. The resulting word is then
\[
w=a^{M_1}\#\cdots\#a^{{\newl}}\#a^{M_{l+1}}\#\cdots\#a^{M_{k+1}}\#b^{{\newb}}c^{M_c}
\]
Most steps in the flow of the analysis below are based on~\cref{lem:cyclethensimplecycle,obs:nonnegativepumptofactorial}. 
We first (partially) handle Form 3. For such $\cV_j$, there is an accepting $\bbN$-run $\rho_j$ on \[\nonnegCol{a^{m^j_1}}\#\cdots\#\nonnegCol{a^{m^j_l}}\#\nonnegCol{a^{m^j_{l+1}}}\#\cdots\#\nonnegCol{a^{m^j_{k+1}}}\#\posCol{b^{m^j_b}}c^{m^j_c}\] 
By pumping $\nonnegative$ cycles as per~\cref{obs:nonnegativepumptofactorial} in all segments by $l$ we obtain an accepting $\bbN$-run $\rho'_j$ on \[\nonnegCol{a^{M_1}}\#\cdots\#\nonnegCol{a^{m^j_l}}\#\nonnegCol{a^{M_{l+1}}}\#\cdots\#\nonnegCol{a^{M_{k+1}}}\#\posCol{b^{m^j_b}}c^{m^j_c}\]
We now pump arbitrary cycles in the $c$ segment to construct a $\bbZ$-run $\rho''_j$ on: \[\nonnegCol{a^{M_1}}\#\cdots\#\nonnegCol{a^{m^j_l}}\#\nonnegCol{a^{M_{l+1}}}\#\cdots\#\nonnegCol{a^{M_{k+1}}}\#\posCol{b^{m^j_b}}c^{M_c}\] 
Next, we compensate for possible negative cycles in the $c$ segment by pumping a $\positive$ cycle in the $b$ segment. Thus, we construct an $\bbN$-run $\rho'''_j$ on \[\nonnegCol{a^{M_1}}\#\cdots\#\nonnegCol{a^{m^j_l}}\#\nonnegCol{a^{M_{l+1}}}\#\cdots\#\nonnegCol{a^{M_{k+1}}}\#\posCol{b^{\newb}}c^{M_c}\] where $\newb$ is chosen to be large enough such that $\rho'''_j$ is indeed an $\bbN$-run for all $1\leq j \leq k$.
Note that it remains to fix $\newl$.

We now turn to Form 1 with a similar process. 
We start with an accepting $\bbN$-run $\rho_j$ on \[\nonnegCol{a^{m^j_1}}\#\cdots\#\posCol{a^{m^j_l}}\#a^{m^j_{l+1}}\#\cdots\#a^{m^j_{k+1}}\#b^{m^j_b}c^{m^j_c}\]
Pump $\nonnegative$ cycles to obtain an accepting $\bbN$-run $\rho'_j$ on  \[\nonnegCol{a^{M_1}}\#\cdots\#\posCol{a^{m^j_l}}\#a^{m^j_{l+1}}\#\cdots\#a^{m^j_{k+1}}\#b^{m^j_b}c^{m^j_c}\] 
obtain a $\bbZ$-run $\rho''_j$ by pumping arbitrary cycles in the remaining segments, including the $b$ segment:
\[\nonnegCol{a^{M_1}}\#\cdots\#\posCol{a^{m^j_l}} \#a^{M_{l+1}}\#\cdots\#a^{M_{k+1}}\#b^{\newb}c^{M_c}\] 
and compensate for negative cycles by taking $\newl$ large enough so that pumping $\positive$ cycles in segment $l$ yields an accepting $\bbN$-run $\rho'''_j$ on \[\nonnegCol{a^{M_1}}\#\cdots\#\posCol{a^{\newl}}\#a^{M_{l+1}}\#\cdots\#a^{M_{k+1}}\#b^{\newb}c^{M_c}\]

We now return to Form 3 and fix the $l$-th segment by pumping $\nonnegative$ cycles to construct an accepting $\bbN$-run on \[\nonnegCol{a^{M_1}}\#\cdots\#\nonnegCol{a^{\newl}}\#\nonnegCol{a^{M_{l+1}}}\#\cdots\#\nonnegCol{a^{M_{k+1}}}\#\posCol{b^{\newb}}c^{M_c}\]

We are left with Form 2, which are the easiest to handle. We simply pump $\nonnegative$ cycles in all segments to construct an accepting $\bbN$-run $\rho'_j$ on  \[\nonnegCol{a^{M_1}}\#\cdots\#\nonnegCol{a^{\newl}}\#\nonnegCol{a^{M_{l+1}}}\#\cdots\#\nonnegCol{a^{M_{k+1}}}\#\nonnegCol{b^{\newb}}\nonnegCol{c^{M_c}}\] 

Note that the requirement for all segments leading up to $l$ to be $\nonnegative$ is crucial, otherwise we would not have been able to pump all the Forms simultaneously.

We now have that $w$ 
is accepted by every $\cV_j$, and segment $l$ is bad for all $\cV_j$ for it.

By applying~\cref{apx:lem:bad_segements_are_bad} for each of the $\cV_j$ and taking global constants to be the products of the respective constants $x,y,z$ for each $\cV_j$, We can now find $X,Y,Z\in \bbN$ multiples of $\alpha$ such that for every $n\in \bbN$ the word 
\[
w_n=a^{M_1}\#\cdots\#a^{{\newl+Xn}}\#a^{M_{l+1}}\#\cdots\#a^{M_{k+1}}\#b^{{\newb+Yn}}c^{M_c+Zn}
\]
is accepted by every $\cV_j$. 
We then choose $n$ large enough to satisfy $\sum_{i\in [k+1]\setminus\{l\}} M_i < \min\{\newb+Yn,M_c+Zn\}$, so that $w_n\notin \lang(\cP)$, since segment $l$ can only pay for $b$ or for $c$, and the remaining segments cannot pay for the remaining segment.

This contradicts the assumption that $\lang(\cP)=\bigcap_{j\in [k]}\lang(\cV_j)$, concluding the proof of~\cref{thm:2CN_prime}.
\hfill\qed

\begin{remark}[Unbounded Compositeness]
\label{apx:rmk:unbounded_compositeness}
The proof of~\cref{thm:2CN_prime} shows that if words with $k+1$ segments are allowed, then the language is not $(1,k)$-composite (and we use it to establish primality). 
By intersecting $\lang(\cP)$ with words that allow at most $k+1$ segments, we obtain a language that is not $(1,k)$-composite, but it is not hard to show that it is $(1,2^{k+1})$-composite. 
This shows that a $2$-CN can be composite, but require unboundedly many factors.
\end{remark}

\section{Regularity of Deterministic Counter Nets is Decidable}
%\label{sec:decidable_DCN}
\label{apx:DCN regularity}
%We now turn our attention to decidable fragments of primality. 
%A natural candidate for decidability is the deterministic fragment. 
%Recall that by~\cref{prop:DCN_Projection}, a $k$-DCN is dimension-minimal if and only if it is not $(1,k-1)$-composite. 
%Thus, dimension-minimality ``captures'' primality, and is our focus. 
%Following, we show that regularity, equivalent to being $(0,1)$-composite, is decidable for $k$-DCN.
%, and being $(1,1)$-composite is decidable for $2$-DCN.

To show that regularity in $k$-DCNs is decidable in \class{EXPSPACE}, we reduce the problem to regularity of Vector Addition Systems (VAS).
In our context, a VAS is a single-state $k$-DCN. 
Given the single state and determinism, we can then associate the letter corresponding to each transition with the transition itself.
Regularity of VAS\footnote{In~{\cite[Equation (11)]{BlockeletS11}}, the authors demonstrate the VAS non-regularity can be expressed in decidable fragment of existential computational tree logic specifically designed for describing properties of VAS coverability graphs.} was shown to be decidable, and is in fact in \class{EXPSPACE}~\cite[Theorem 4.5]{BlockeletS11}. 

\begin{remark}
\label{rmk:regularity_VASS}
Regularity of VAS with states (VASS) was shown to be decidable in~\cite{Demri13}. 
However, the model there does not have accepting states. 
It appears to be possible to extend the techniques used there to our setting, but would this require re-proving several results. 
Therefore, we take the following approach through VAS.
%%%Additionally, we note that a stronger result appears in an unpublished manuscript (at the time of writing of this paper)~\cite{bose2023history}.
\end{remark}

Throughout this section, we consider a $k$-DCN $\cD=\tup{\Sigma,Q,Q_0,\delta,F}$. 
Our construction proceeds as follows. 
First, we obtain from $\cD$ a $k$-DCN $\cC$ with the same structure, but such that its transitions are distinctly labelled. 
Next, we invoke a technique of Hopcroft and Pansiot~\cite{HopcroftP79}; we will construct a single-state $(k+3)$-DCN (namely a VAS) $\cU$ whose language ``approximates'' that of $\cC$, and in particular preserves regularity.	
Finally, we invoke the deciability of regularity for deterministic VAS~\cite{BlockeletS11}.	
The crux of the argument is to maintain regularity throughout, even though the language changes in each step.	
In the following, a \emph{DFA} (resp. NFA) is a $0$-DCN (resp. $0$-CN).
\begin{lemma}
    \label{lem:DCN_to_relabelled}
    % Given a $k$-DCN $\cD=\tup{\Sigma,Q,Q_0,\delta,F}$, define $\cC = \tup{ \Gamma, Q, Q_0, \delta', F}$ where $\Gamma = \set{ \gamma_t \mid t \in \delta}$ and $\delta' = \set{ (p, \gamma_t, \vec{v}, q) \mid t = (p, \sigma, \vec{v}, q) \in \delta)}$. 
    % Then $\lang(\cD)$ is regular if and only if $\lang(\cC)$ is regular.

    	Let $\cD$ be the above $k$-DCN and define $\cC=\tup{\Gamma,Q,Q_0,\delta',F}$ where $\Gamma=\{\gamma_t\mid t\in\delta\}$ with $(p,\gamma_t,\vec{v},q)\in \delta'$ if and only if $t=(p,\sigma,\vec{v},q)\in \delta$. Then $\lang(\cD)$ is regular if and only if $\lang(\cC)$ is regular.
\end{lemma}
\begin{proof}
    For the first direction, assume $\lang(\cC)$ is regular, we show that $\lang(\cD)$ is regular.
    Let $\cA$ be a DFA such that $\lang(\cA) = \lang(\cC)$. 
    Now consider the process of reverting the labelling, to construct an NFA $\cA'$ with the same state space and same transition structure, but different labels.
    For a transition $(s, \gamma_t, t)$ in $\cA$, let $t = (p, \sigma, \vec{v}, q)$ be the transition in $\cD$ corresponding to $\gamma_t$, then we introduce in $\cA'$ the transition $(s, \sigma, t)$. Note that $\cA'$ may be nondeterministic.
    
    Now, since $\lang(\cA) = \lang(\cC)$, we claim that $\lang(\cA') = \lang(\cD)$. 
    Consider $w \in \lang(\cA')$ with $w=\sigma_1\cdots \sigma_n$, then there exist transitions $t_1,\ldots, t_n$ where each $t_i$ is labelled $\sigma_i$ such that $\gamma_{t_1}\cdots \gamma_{t_n}\in \lang(\cA)=\lang(\cC)$. The sequence $t_1,\ldots, t_n$ represents an accepting run of $\cD$, which is labelled by $w = \sigma_1 \cdots \sigma_n$, so $w\in \lang(\cD)$.
    Conversely, consider $w\in \lang(\cD)$ with $w = \sigma_1 \cdots \sigma_n$, then the sequence of transitions in the accepting run of $\cD$ on $w$ is $t_1,\ldots, t_n$ such that $\gamma_{t_1}\cdots \gamma_{t_n}\in \lang(\cC)=\lang(\cA)$, so by the construction of $\cA'$ we have that $w=\sigma_1\cdots \sigma_n\in \lang(\cA')$. We conclude that if $\lang(\cC)$ is regular, then so is $\lang(\cD)$.

    For the second direction, assume $\lang(\cD)$ is regular, we show that $\lang(\cC)$ is regular. 
    Let $\cA$ be a DFA such that $\lang(\cA) = \lang(\cD)$. 
    We consider a product DFA $\cA \times \cD$ over the alphabet $\Gamma$ as follows.
    When reading letter $\gamma_t$, for $t = (p, \sigma, \vec{v}, q)$, it transitions from state $(r, p)$ to $(r',q)$, where $r'$ is reached when reading $\sigma$ from $r$ in $\cA$ and where $q \in Q$ is reached when taking transition $t$ from $p \in Q$.  
    Acceptance is determined by $\cA$. 
    Intuitively, $\cA \times \cD$ simulates $\cA$ while ensuring that the transitions being read actually form a $\bbZ$-run of $\cD$. 

    We claim that $\lang(\cA\times \cD)=\lang(\cC)$. 
    Indeed, $\cC$ accepts a sequence $\gamma_{t_1} \cdots \gamma_{t_n}$ if and only if $t_1 \cdots t_n$ forms an accepting run in $\cD$, if and only if $t_1 \cdots t_n$ forms a $\bbZ$-run of $\cD$ and the letters on the transitions are $\sigma_1\cdots\sigma_n\in \lang(\cD)=\lang(\cA)$. 
    Note that this last double implication relies on the fact that $\cD$ is deterministic, hence there is a single run of $\cD$ on any given word.
\hfill\qed\end{proof}

We now employ a well-known technique of Hopcroft and Pansiot~{\cite[Lemma 2.1]{HopcroftP79}} that converts a $k$-VASS to a $(k+3)$-VAS by simulating the finite control states using three extra dimensions.
The first $k$ dimensions just mirror the $k$ dimensions in the original VASS and the last three dimensions operate the state simulation; one of the three extra dimensions explicitly maintains a value corresponding to the current state.
In the unlabelled (VASS) setting, the construction takes a transition $(p, \vec{x}, q)$ and spawns three transitions $\vec{v}_1$, $\vec{v}_2$, and $\vec{v}_3$ in the corresponding VAS.
Crucially, these vectors are chosen in such a way that if $(p, \vec{x}, q)$ is taken in the original VASS, then $\vec{v}_1, \vec{v}_2, \vec{v}_3$ can be taken in sequence.
Moreover, if $\vec{v}_1$ is taken, then the only transition that could follow is $\vec{v}_2$, and then the only transition that could follow is $\vec{v}_3$; ending with the last three dimensions taking values corresponding to $q$.

Following, in Lemma~\ref{lem:relabelled_to_VAS}, we show that one can replicate this construction starting with the distinctly-labelled $k$-DCN $\cC$ and obtaining a single-state $(k+3)$-DCN $\cU$.
In particular, if a $\sigma$ labelled transition, $(p, \sigma, \vec{v}, q)$, is present in the $k$-DCN $\cC$, then labelled transitions $(\sigma_1, \vec{v}_1)$, $(\sigma_2,\vec{v}_2)$, and $(\sigma_3,\vec{v}_3)$ are present in the constructed $(k+3)$-DCN $\cU$.

\begin{lemma}[Corollary of~{\cite[Lemma 2.1]{HopcroftP79}}]
    \label{lem:relabelled_to_VAS}
    Let $\cC$ be a distinctly-labelled $k$-DCN over alphabet $\Gamma$ (obtained as per~\cref{lem:DCN_to_relabelled}), then there exists a single-state $(k+3)$-DCN $\cU$ over \mbox{alphabet} $\Upsilon = \set{ \gamma_1, \gamma_2, \gamma_3 \mid \gamma \in \Gamma }$ such that 
    $\lang(\cU) = \{
        a_1 a_2 a_3 \, b_1 b_2 b_3 \cdots c_1,$ 
        $a_1 a_2 a_3 \, b_1 b_2 b_3 \cdots c_1c_2,$ 
        $a_1 a_2 a_3 \, b_1 b_2 b_3  \cdots c_1c_2c_3$ 
        $\mid  a \, b \cdots c \in \lang(\cC) \wedge a, b, \ldots, c \in \Gamma
    \}$.
     %\begin{equation*}
    % $\lang(\cU) = 
    % \left\{ 
    % \begin{tabular}{l|c}
    %      $a_1 a_2 a_3 \, b_1 b_2 b_3  \,\cdots\, c_1,$ & \\
    %      $a_1 a_2 a_3 \, b_1 b_2 b_3  \,\cdots\, c_1 c_2,$  & $ a \, b \,\cdots\, c \in \lang(\cC) \wedge a, b, \ldots, c \in \Gamma$ \\
    %      $a_1 a_2 a_3 \, b_1 b_2 b_3  \,\cdots\, c_1 c_2 c_3$ &
    % \end{tabular}       
    % \right\}.
    %\end{equation*}
\end{lemma}
    \begin{proof}
    We give the construction of the VAS, this proof only differs from the proof given by Hopcroft and Pansoit~\cite{HopcroftP79} by the inclusion of transition labels.
    Assume that $\cC$ has $n$ states $Q = \{q_1, \ldots, q_n\}$.
    For each $i \in \set{1, \ldots, n}$, let $a_i = i$ and $b_i = (n+1)(n+1-i)$.
    For each transition $t = (q_i, \sigma, \vec{x}, q_j)$ in $\cC$, there are three transitions $t_1 = \sigma_1,(\vec{0}, -a_i, a_{n+1-i}-b_i, b_{n+1-i})$, $t_2 = \sigma_2,(\vec{0}, b_i, -a_{n+1-i})$, and $t_3 = \sigma_3,(\vec{x}, a_j - b_i, b_j, -a_i)$ in $\cU$ (we omit the states from the transitions, since there is only one state). 
    
    Observe that since each transition in $\cC$ has a distinct label, then each of the transitions in $\cU$ also has a distinct label, making $\cU$ a DCN.
    It remains to show that $\cU$ recognises the language given in the statement of this lemma.
    Suppose the word $\sigma_1\,\sigma_2\,\cdots\,\sigma_k \in \lang(\cC)$, then there is a path $\pi = (t_i)_{i=1}^k$ in $\cC$, so $t_i = (q_i, \sigma_i, \vec{x}_i, q_{i+1})$.
    Since $\Uu$ faithfully simulates $\cC$, the path $\pi' = (t_{i,1}\, t_{i,2}\, t_{i,3})_{i=1}^k$ in $\Uu$ witnesses the word $\sigma_{1,1} \,\,\! \sigma_{1,2} \,\,\!\sigma_{1,3} \; \sigma_{2,1} \,\,\! \sigma_{2,2} \,\,\! \sigma_{2,3} \,\cdots\, \sigma_{k,1} \,\,\! \sigma_{k,2} \,\,\! \sigma_{k,3}$. 
    In fact, the last three transitions need not all be executed; indeed $\sigma_{1,1} \,\,\! \sigma_{1,2} \,\,\! \sigma_{1,3} \,\cdots\, \sigma_{k,1} \in \lang(\Uu)$, $\sigma_{1,1} \sigma_{1,2} \sigma_{1,3} \,\cdots\, \sigma_{k,1} \,\,\! \sigma_{k,2} \in \lang(\Uu)$, and $\sigma_{1,1} \,\,\! \sigma_{1,2} \,\,\! \sigma_{1,3} \,\cdots\, \sigma_{k,1} \,\,\! \sigma_{k,2} \,\,\! \sigma_{k,3} \in \lang(\cU)$.
    \hfill\qed\end{proof}
\begin{proposition}
\label{pro:triple-letter-regularity}
    Let $L \subseteq \set{ a, b, \ldots, c }^*$ and $L' \subseteq \set{ a_1, a_2, a_3, b_1, b_2, b_3, \ldots, c_1, c_2, c_3 }^*$ such that 
    $L' = \{
        a_1 a_2 a_3 \, b_1 b_2 b_3 \cdots c_1,$ 
        $a_1 a_2 a_3 \, b_1 b_2 b_3 \cdots c_1c_2,$ 
        $a_1 a_2 a_3 \, b_1 b_2 b_3  \cdots c_1c_2c_3$ 
        $\mid  a \, b \,\cdots\, c \in L
    \}$,
    then $L$ is regular if and only if $L'$ is regular.
    % \begin{equation*}
    % L' = 
    % \left\{ 
    % \begin{tabular}{l|c}
    %      $a_1 a_2 a_3 \, b_1 b_2 b_3  \,\cdots\, c_1,$ & \\
    %      $a_1 a_2 a_3 \, b_1 b_2 b_3  \,\cdots\, c_1 c_2,$  & $a \, b \,\cdots\, c \in L$ \\
    %      $a_1 a_2 a_3 \, b_1 b_2 b_3  \,\cdots\, c_1 c_2 c_3$ &
    % \end{tabular}       
    % \right\}.
    % \end{equation*}
\end{proposition}
\begin{proof}
    This is a basic exercise in automata. The direction $L$ is regular $\implies$ $L'$ is regular is proved by constructing a DFA for $L'$ that verifies letters are read in triplets (i.e., $a_1a_2a_3$, possibly stopping in the middle), and after every triplet simulates the corresponding transition in a DFA for $L$.

    The direction $L'$ is regular $\implies$ $L$ is regular is proved by constructing an NFA for $L$ that given letter $a$ simulates the transition of a DFA for $L'$ on $a_1a_2a_3$.
\hfill\qed\end{proof}

% \begin{proof}
%     This is a basic exercise in automata. The direction $L$ is regular $\implies$ $L'$ is regular is proved by constructing a DFA for $L'$ that verifies letters are read in triplets (i.e., $a_1a_2a_3$, possibly stopping in the middle), and after every triplet simulates the corresponding transition in a DFA for $L$.

%     The direction $L'$ is regular $\implies$ $L$ is regular is proved by constructing an NFA for $L$ that given letter $a$ simulates the transition of a DFA for $L'$ on $a_1a_2a_3$.
% \hfill\qed\end{proof}

% \begin{theorem}[cf.~{\cite[Theorem 4.5]{BlockeletS11}}]
%     \label{thm:dvas-regularity}
%     DVAS regularity is decidable and is in EXPSPACE.
% \end{theorem}

By combining~\cref{lem:DCN_to_relabelled},~\cref{lem:relabelled_to_VAS},~\cref{pro:triple-letter-regularity}, and~\cite[Theorem 4.5]{BlockeletS11}, we can conclude with the proof of \cref{thm:DCN_regularity_decidable}.
   
    \paragraph{Proof of~\cref{thm:DCN_regularity_decidable}.}
    Given a $k$-CN $\cD$, construct a distinctly-labelled $k$-DCN $\cC$ as per~\cref{lem:DCN_to_relabelled}.
    Next, apply~\cref{lem:relabelled_to_VAS} to obtain a single-state, distinctly-labelled $k+3$-DVAS $\cU$.
    Finally, treat $\cU$ as a VAS. 
    Indeed, regularity of VAS is with respect to the transition ``names'', which are unique in $\cU$ due to the distinct labels. 
    We can then decide the regularity of $\cU$ using~{\cite[Theorem~4.5]{BlockeletS11}}. 

    We have that $\cU$ is regular if and only if $\lang(\cD)$ is regular by~\cref{lem:DCN_to_relabelled,lem:relabelled_to_VAS,pro:triple-letter-regularity}. Moreover, our construction can clearly be implemented in polynomial time, thus giving the same complexity bound.\qed 

% that regularity is decidable for $k$-DCNs.
% \begin{theorem}
%     \label{thm:DCN_regularity_decidable}
%     Regularity of $k$-DCN is decidable and is in \class{EXPSPACE}.
% \end{theorem}

% \begin{proof}
%     Given a $k$-CN $\cD$, construct a distinctly-labelled $k$-DCN $\cC$ as per~\cref{lem:DCN_to_relabelled}.
%     Next, apply~\cref{lem:relabelled_to_VAS} to obtain a single-state, distinctly-labelled $k+3$-DVAS $\cU$.
%     Finally, treat $\cU$ as a VAS. 
%     Indeed, regularity of VAS is with respect to the transition ``names'', which are unique in $\cU$ due to the distinct labels. 
%     We can then decide the regularity of $\cU$ using \cite[Theorem 4.5]{BlockeletS11}. 

%     We have that $\cU$ is regular iff $\lang(\cD)$ is regular by~\cref{lem:DCN_to_relabelled,lem:relabelled_to_VAS,pro:triple-letter-regularity}. Moreover, our construction can clearly be implemented in polynomial time, thus giving the same complexity bound.
% \hfill\qed\end{proof}

\end{document}